\documentclass[12pt, draftclsnofoot, onecolumn]{IEEEtran}

\usepackage{ifpdf}
\usepackage{cite}
\usepackage[pdftex]{graphicx}

\usepackage[T1]{fontenc}
\usepackage{enumitem}
\usepackage{afterpage}
\usepackage{bm}
\usepackage{listings}
\usepackage{xcolor}
\usepackage{graphicx}
\usepackage{float}
\usepackage[font=small]{caption}
\usepackage[font=footnotesize]{subcaption}

\usepackage{placeins}
\usepackage{url}
\usepackage{epstopdf}
\usepackage{booktabs}
\usepackage{footnote}
\makesavenoteenv{table}
\makesavenoteenv{tabular}
\usepackage{flushend}
\usepackage{amsmath}
\usepackage{amssymb}
\usepackage{tabularx}
\usepackage{booktabs}
\usepackage{placeins}
\usepackage{graphicx}
\usepackage{color}
\usepackage{colortbl}
\usepackage{bbold}

\usepackage[acronyms,nonumberlist,nopostdot,nomain,nogroupskip]{glossaries}

\usepackage{tikz}
    \usetikzlibrary{shapes.arrows}

\usepackage{pgfplots}
\usepackage{adjustbox}

\usepackage{gincltex}
\usepgfplotslibrary{fillbetween}
\usetikzlibrary{plotmarks}
\usetikzlibrary{patterns}
\usepackage{mathtools}
\usepackage{booktabs}
\usepackage{threeparttable}
\usepackage{amsfonts}
\DeclareMathSymbol{\shortminus}{\mathbin}{AMSa}{"39}
\usepackage{amsthm}

\newlength{\hatchspread}
\newlength{\hatchthickness}
\newlength{\hatchshift}
\newcommand{\hatchcolor}{}
\tikzset{hatchspread/.code={\setlength{\hatchspread}{#1}},
         hatchthickness/.code={\setlength{\hatchthickness}{#1}},
         hatchshift/.code={\setlength{\hatchshift}{#1}},% must be >= 0
         hatchcolor/.code={\renewcommand{\hatchcolor}{#1}}}
\tikzset{hatchspread=3pt,
         hatchthickness=0.8pt,
         hatchshift=0pt,% must be >= 0
         hatchcolor={blue!20}}
\pgfdeclarepatternformonly[\hatchspread,\hatchthickness,\hatchshift,\hatchcolor]% variables
   {custom north west lines}% name
   {\pgfqpoint{\dimexpr-2\hatchthickness}{\dimexpr-2\hatchthickness}}% lower left corner
   {\pgfqpoint{\dimexpr\hatchspread+2\hatchthickness}{\dimexpr\hatchspread+2\hatchthickness}}% upper right corner
   {\pgfqpoint{\dimexpr\hatchspread}{\dimexpr\hatchspread}}% tile size
   {% shape description
    \pgfsetlinewidth{\hatchthickness}
    \pgfpathmoveto{\pgfqpoint{0pt}{\dimexpr\hatchspread+\hatchshift}}
    \pgfpathlineto{\pgfqpoint{\dimexpr\hatchspread+0.15pt+\hatchshift}{-0.15pt}}
    \ifdim \hatchshift > 0pt
      \pgfpathmoveto{\pgfqpoint{0pt}{\hatchshift}}
      \pgfpathlineto{\pgfqpoint{\dimexpr0.15pt+\hatchshift}{-0.15pt}}
    \fi
    \pgfsetstrokecolor{\hatchcolor}
    \pgfusepath{stroke}
   }

\newtheorem{theorem}{Theorem}

\newtheorem{corollary}{Corollary}[theorem]
\newtheorem{definition}{Definition}

\def \fwidth {0.6\columnwidth}
\def \fheight {0.3\columnwidth}
\def \subfwidth {0.8\linewidth}
\def \subfheight {0.6\linewidth}

\definecolor{violet}{rgb}{0.6,0,0.6}%
\definecolor{orange_D}{rgb}{1,0.3,0}%
\definecolor{cyan}{rgb}{0,0.67,0.64}%
\definecolor{red}{rgb}{0.9,0,0}%
\definecolor{green}{rgb}{0,0.8,0}%
\definecolor{yellow}{rgb}{1,0.8,0}

\usepackage{glossaries}
\newacronym{aoi}{AoI}{Age of Information}
\newacronym{paoi}{PAoI}{Peak AoI}
\newacronym{pdf}{PDF}{Probability Density Function}
\newacronym{cdf}{CDF}{Cumulative Density Function}
\newacronym{fcfs}{FCFS}{First Come First Serve}
\newacronym{lcfs}{LCFS}{Last Come First Serve}
\newacronym{iot}{IoT}{Internet of Things}
\newacronym{mdp}{MDP}{Markov Decision Process}
\newacronym{csma}{CSMA}{Carrier Sense Multiple Access}
\newacronym{urllc}{URLLC}{Ultra Reliable Low Latency Communication}
\newacronym{3cls}{3CLS}{Communications, Computing,  Control,  Localization,  and  Sensing}
\newacronym{mps}{MPS}{Multi-Purpose 3CLS and energy Services}
\newacronym{mpr}{MPR}{Multi-Packet Reception}
\newacronym{unb}{UNB}{Ultra-Narrowband}
\newacronym{bs}{BS}{Base Station}

\newcommand{\edit}[1]{\textcolor{black}{#1}}

\pgfplotsset{compat=1.15}

\begin{document}

\title{Peak Age of Information Distribution for Edge Computing with Wireless Links}

\author{Federico Chiariotti, Olga Vikhrova, Beatriz Soret, Petar Popovski\thanks{F. Chiariotti (corresponding author, email: fchi@es.aau.dk), B. Soret, and P. Popovski are with the Department of Electronic Systems, Aalborg University, Denmark. O. Vikhrova is with the DIIES Department, University Mediterranea of Reggio Calabria, Italy.}
}

\maketitle

\begin{abstract}
\gls{aoi} is a critical metric for several \gls{iot} applications, where sensors keep track of the environment by sending updates that need to be as fresh as possible. The development of edge computing solutions has moved the monitoring process closer to the sensor, reducing the communication delays, but the processing time of the edge node needs to be taken into account. %Furthermore, system designers need to know the full distribution of the \gls{paoi} to plan for the worst-case scenario, i.e., the one in which 
%the age of the measurement is highest. 
\edit{Furthermore, a reliable system design in terms of freshness requires the knowledge of the full distribution of the \gls{paoi}, from which the probability of occurrence of rare, but extremely damaging events can be obtained.}
In this work, we model the communication and computation delay of such a system 
 as two \gls{fcfs} queues in tandem, analytically deriving the full distribution of the \gls{paoi} for the $M/M/1$ -- $M/D/1$ and the $M/M/1$ -- $M/M/1$ tandems, which can represent a wide variety of realistic scenarios.
\end{abstract}

\glsresetall


\section{Introduction}\label{sec:intro}
Traditional communication networks consider packet delay as the one and only performance metric to capture the latency requirements of a transmission. However, numerous \gls{iot} applications require the transmission of real-time status updates of a process from a generating point to a remote destination~\cite{abd2019role}. Sensor networks, vehicular networks and other tracking systems, and industrial control are examples of this kind of update process. For these cases, the \gls{aoi} is a novel concept that better represents timeliness requirements by quantifying the freshness of the information at the receiver \cite{kaul2012real}. Basically, \gls{aoi} computes the time elapsed since the latest received update was generated at any given moment in time, i.e., how old is the last packet received by the destination. Another age-related metric is the \gls{paoi}, which is the maximum value of \gls{aoi} for each update, i.e., how old the last packet was when the next one is received by the destination. As in other performance metrics of communication systems, the \gls{paoi} is more informative than the average age when the interest is in worst-case analysis, e.g., when the system requirement is on the tail of the distribution. To illustrate, the \gls{paoi} can be useful to limit the latency of networked control systems, ensuring that the receiver has a recent picture of the state of the transmitter.

Edge computing is a technology that is gaining traction in age-sensitive \gls{iot} applications, as the transmission of sensor readings to a centralized cloud requires too much time and increases uncertainty, while processing the received data closer to the sensor that generated them can reduce the communication latency and the overall age of the information available to the monitoring or control process~\cite{yu2017survey}. 
%The combination of edge computing capabilities with deep learning techniques is a promising research and development direction for the \gls{iot} in the coming years~\cite{li2018learning}. 
\edit{Freshness of information is going to be one of the critical parameters in beyond 5G and 6G, enabling \gls{3cls} services~\cite{saad2019vision}: these applications require joint sensing, computation, and communication resources. In particular, %\gls{mps} 
services in telehealth, agriculture, manufacturing plants and robotics require strict control performance guarantees, which can only be possible by carefully designing the communication system. Furthermore, \gls{aoi} and \gls{paoi} are often the relevant timing metrics for control systems, as they represent the time difference between the system observed by the controller and the real one when the control action takes place. As these services require high reliability, analyzing the average age is not enough: the tail of its distribution is also a very important parameter, as it directly affects the risk of control system failures, which must be constrained to very low levels for critical applications.}
In these edge applications that combine communication and computation, the contribution of both to the \gls{aoi} needs to be taken into account. Limiting the age of the processed information is a critical requirement, which can influence the choice between local and edge-based computation~\cite{kuang2019age} for \gls{iot} nodes. The problem becomes even more complex when considering multiple sources and different packet generation behaviors, along with limited communication capabilities~\cite{li2019general}.

\begin{figure}[!t]
 \centering
  \includegraphics[width=0.45\columnwidth]{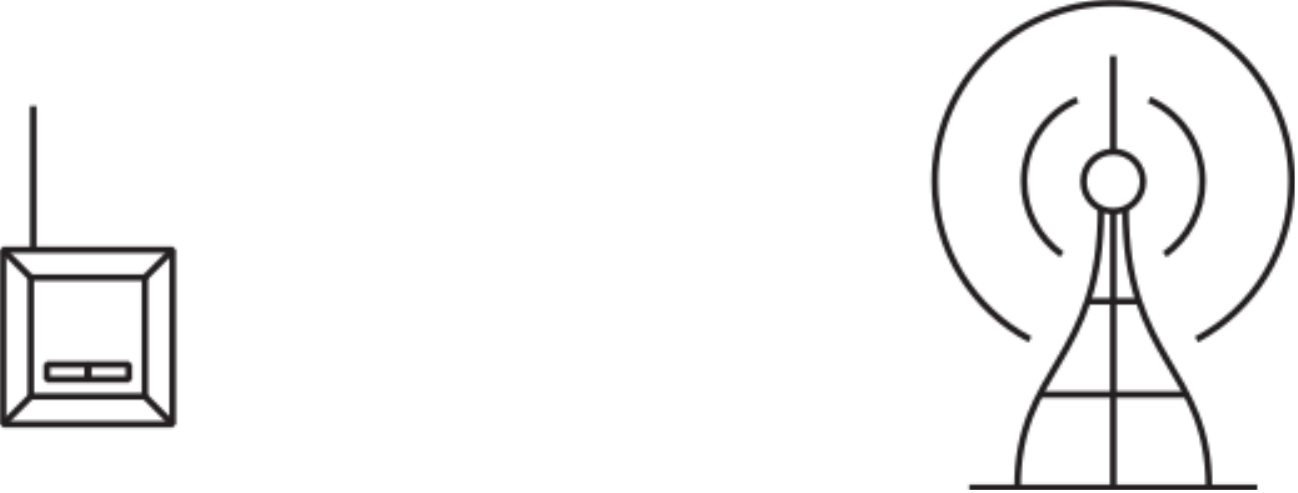}
  \vspace{0.3cm}
  
 \includegraphics{./tandem.tex}
 \vspace{-0.1cm}
 \caption{An example of the \gls{iot} edge computing use case: the sensor transmits data to a computing-enabled edge node, which needs to maintain information freshness.}
 \vspace{-0.45cm}
 \label{fig:edge}
\end{figure}

Tandem queues (specifically the minimal 2-nodes case) are then a natural modeling choice for this kind of scenarios, as the first system represents the communication link, while the second one represents the computing-enabled edge node and its task queue. \edit{Fig.~\ref{fig:edge} shows an example: the communication buffer at the sensor and the task queue on the edge computing-enabled \gls{bs} are the two queues in the tandem.}

If the load on the computing-enabled edge node is time constant, while communication is less predictable due to, e.g., dynamic channel variations and random access, then the $M/M/1$ -- $M/D/1$ tandem queue is an appropriate model. \edit{An $M/M/1$ queue is a proper model for the channel if we assume an ALOHA system~\cite{Norman1970} with perfect \gls{mpr}, in which~\cite{Goseling2015} the packets are not lost due to collisions and can only be lost due to channel errors. This model is suitable for IoT systems based on \gls{unb} transmissions, such as SigFox \cite{mroue2018mac}. On the other hand, computation time is often modeled as a linear function of the data size in the literature~\cite{song2019age}, and updates of the same size would have a constant and deterministic service time, therefore well-represented by an $M/D/1$ queue.} Another option is to consider the computing load as also variable over time, leading to stochastic computing time, and represent the two systems by $M/M/1$ queues with different service rates~\cite{green1980queueing}.

%The minimal tandem queue is the 2-node case in
An example of the \gls{aoi} dynamics is plotted in Fig.~\ref{fig:aoi_triangles}: the \gls{aoi} grows linearly over time, then decreases instantly when a new update arrives. Naturally, the age at the destination is never lower than the age at the intermediate node, as each packet that reaches the destination has already passed through the intermediate node, but the dynamics between the two are not trivial and serve as a motivation for this work. We analyze the distribution of the \gls{paoi} in a tandem queue with two systems with independent service times and a single source, where each infinite queue follows the \gls{fcfs} policy. We consider both the $M/M/1$ -- $M/D/1$ tandem and the $M/M/1$ -- $M/M/1$ case, covering common communication relaying and communication and computation scenarios. Our aim is to derive the complete \gls{pdf} of the \gls{paoi} for systems with arbitrary packet generation and service rates, allowing system designers to define reliability requirements using \gls{paoi} thresholds and derive the network specifications needed to meet those requirements.

\begin{figure}[!t]
\centering
 \includegraphics{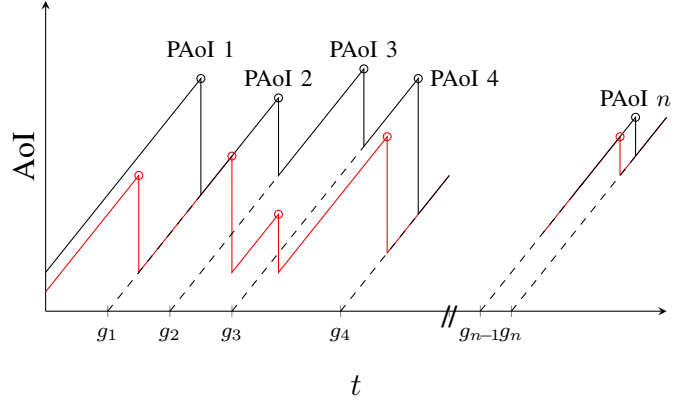}
 \vspace{-0.3cm}
 \caption{The time-evolution of the Age of Information in a 2-node tandem queue. The age at the destination is plotted in black, the age at the intermediate node is plotted in red. Departures from
 each node are marked by a circle.}
   \vspace{-0.45cm}
 \label{fig:aoi_triangles}
\end{figure}

Although the motivation for our analysis is the IoT edge computing, we notice there are other age-sensitive applications that can use the same models and results of this paper. For instance, a tandem queue can model a relay networking in which a packet is transmitted through one or more buffer-aided intermediate nodes between transmitter and receiver to, e.g., overcome the physical distance between the two end-points. A good example is a satellite relay that connects the ground and a satellite (or vice versa) through another satellite. In this case, the links are highly unpredictable and depend on different factors such as positioning jitter~\cite{toyoshima2002optimum} and conditions in the upper atmosphere. A tandem queue in which the servers represent the transmission over successive links can represent this kind of system, provided that the service (transmission) times of different links are independent. Blockchain is another application where \gls{aoi} is a critical metric to validate transactions in real time, particularly when combining the use of the distributed ledger with \gls{iot} applications~\cite{Lee2020blockchain}. As nodes need to transmit information, which is then validated, the tandem model is a useful abstraction for the communication and computation time~\cite{rovira2019optimizing}.

\edit{
The contribution of this paper can be summarized by the following points:
\begin{itemize}
 \item The full \gls{pdf} of the \gls{paoi} is derived for the tandem of an $M/M/1$ and an $M/D/1$ queue, which is relevant for \gls{iot} edge computing scenarios;
 \item The same derivation is performed for the tandem of two $M/M/1$ queues, which can be relevant in relay-based communications with two independent links;
 \item Design considerations are drawn for the two systems, using the analytical formulas as a basis for system optimization.
\end{itemize}
}

The structure of the paper is as follows.  In Sec.~\ref{sec:system_model} the system model is detailed, as well as the procedure to calculate the \gls{aoi}. Sec.~\ref{sec:mm1} presents the calculations based on the model for the $M/M/1$ -- $M/M/1$ tandem, while Sec.~\ref{sec:md1} does the same for the $M/M/1$ -- $M/D/1$ tandem. Numerical results are plotted in Sec.~\ref{sec:results}, and the paper is concluded in Sec.~\ref{sec:conclusions}.\footnote{The code used in the simulations, with the implementation of all the theoretical derivations in the paper, is available at https://github.com/AAU-CNT/tandem\_aoi.}

\section{Related work}

An overview of edge computing in IoT can be found in \cite{yu2017survey} and \cite{Hassan2018edge}. The initial research approach to latency in the edge paradigm was, in the context of 5G systems, to understand its potential to support \gls{urllc} \cite{Elbamby2018edge,Elbamby2019wireless,Hu2018delay}. The relevance of the \gls{aoi} for edge computing applications was first identified in \cite{kuang2019age}, although only the average \gls{aoi} is computed. The same authors propose in \cite{gong2020} a joint transmission and computing scheduling for a deadline. Another research area is the use of machine learning, particularly deep learning techniques, to unleash the full potential of \gls{iot} edge computing and enable a wider range of application scenarios~\cite{Wang2020convergence,li2018learning}. 
%Particularly, the combination of edge computing capabilities with deep learning techniques is a promising research and development direction for the \gls{iot} in the coming years}.

\gls{aoi} is a relatively new metric in networking, but it has gained widespread recognition thanks to its relevance to several applications. Most theoretical results refer to simple queuing systems with a single node and \gls{fcfs} policy. However, a few recent works focus their attention in the study of the age in tandem queues, even with different policies such as \gls{lcfs}~\cite{bedewy2017age}.  Some \gls{iot} scenarios have also been modeled as tandem $M/M/1$, $M/D/1$, or $D/M/1$ queues with multiple sources, as the read from the sensor is first pre-processed, and then transmitted to the server~\cite{xu2019optimizing}. In this case, each queue follows the \gls{fcfs} discipline, but the authors derive only the average \gls{paoi}. Another possibility is queue replacement, in which only the freshest update for each source is kept in the queue, reducing queue size significantly: in this case, the replaced packet is not placed in the queue, but dropped altogether, reducing channel usage with respect to simple \gls{lcfs}, with or without preemption. In this case, the queue is modeled as an $M/M/1/2$, and if a new packet arrives it takes the queued packet's place. Some preliminary results on such a system are given in~\cite{pappas2015age}, while the average \gls{aoi} and \gls{paoi} are computed in~\cite{kosta2019queue} for one source and in~\cite{kosta2019age} for multiple sources. An analysis of the effect of preemption on this kind of models on the average \gls{aoi} is presented in~\cite{yates2018age}, and \cite{kam2018age} derives the average \gls{aoi} for two queues in tandem with preemption and different arrival processes. \edit{Another work~\cite{akar2020finding} derives the full distribution of the \gls{paoi} for preemptive queues whose service time follows a phase-type distribution, which can be used to represent multiple hops.} The \gls{pdf} of the \gls{aoi} in multi-hop networks with packet preemption has been derived in~\cite{Ayan2020probability}. Another work~\cite{buyukates2019age} deals with multicast networks in which a single source updates multiple receivers over several hops, deriving the average age in that context. Another recent work considers a scenario in which packets that are not received within a deadline are dropped, deriving the average \gls{paoi}~\cite{li2020age}. \edit{Some of these works deal with multi-hop queuing networks, of which our tandem model is a specific case, but they all assume some form of preemption. The computation of \gls{aoi} and \gls{paoi} in systems with preemption is much simpler, particularly in systems of $M/M/1$ queues, but it cannot represent all the relevant use cases: preemption might not be possible or desirable, depending on the specific requirements of the control or monitoring application. For example, telehealth applications might benefit from receiving even out of date samples. To the best of our knowledge, this work is the first to derive the complete \gls{pdf} of the age for non-preemptive tandem networks, covering these applications. Another example is the satellite relay scenario, where the traffic entering the tandem system comes from an aggregation of devices and therefore each individual update is relevant}. 

Other works concentrate on more realistic models, considering the effect of physical and medium access issues on the \gls{aoi}. A model considering a fading wireless channel with retransmissions was used to compute the \gls{paoi} distribution over a single-hop link in~\cite{devassy2019reliable}, and a recent live \gls{aoi} measurement study on a public networks generally confirmed that the theoretical models are realistic~\cite{beytur2019measuring}. Other works compute the average \gls{aoi} in \gls{csma}~\cite{maatouk2020age}, ALOHA~\cite{yates2020age} and slotted ALOHA~\cite{chen2020age} networks, considering the impact of the different medium access policies on the age. \edit{It is also possible to jointly optimize the sampling and updating processes, i.e., both the reading instants from the sensor and the transmission of updates, if both are controllable~\cite{li2020proc}: the cost of both operations is a determinant of the overall \gls{aoi} of the system~\cite{zhou2019joint}.}

A generalization of our model, relevant for applications like satellite relaying, is the multi-hop network with and arbitrary number of $M/M/1$ systems, senders, and receivers: in this case, the moment generating function of the \gls{aoi} was derived in~\cite{yates2020agenet} for preemptive servers. Other scheduling policies make things more complicated: tight bounds for the average \gls{aoi} were derived in~\cite{chiariotti2020information} for the \gls{fcfs} discipline and other \gls{aoi}-oriented queue prioritization mechanisms.

Deriving the complete distribution of the age might be critical for reliability-oriented applications, but it is still mostly unexplored in the literature: while the work on deriving the first moments of the \gls{aoi} distribution is extensive, the analytical complexity of deriving the complete \gls{pdf} is a daunting obstacle. A recent work~\cite{champati2019statistical} uses the Chernoff bound to derive an upper bound of the quantile function of the \gls{aoi} for two queues in tandem with deterministic arrivals, but to the best of our knowledge, the complete \gls{paoi} distribution has only been derived for simple queuing systems~\cite{devassy2019reliable}. \edit{Another interesting approach~to achieve reliability is to consider the \gls{aoi} process directly: in~\cite{liu2019taming}, the authors use extreme value theory to derive the complementary \gls{cdf} of the \gls{paoi} in a realistic channel setting with scheduled transmissions. \cite{chaccour2020} proposes the use of the risk of ruin, an economic concept,  as a metric of fresh and reliable information in augmented reality. Specifically, the \gls{cdf} of the \gls{paoi} is used to find the probability of maximum severity of ruin \gls{paoi} in single node systems. Yet another approach to reliability, which does not derive the complete distribution of the \gls{aoi} but uses an average measure on its tail, is presented in~\cite{zhou2020risk}: the authors pose the risk minimization problem as a \gls{mdp}, optimizing the age with bounded risk.} For a more complete overview of the literature on \gls{aoi}, we refer the reader to~\cite{yates2020agesurvey}.

\section{System model} \label{sec:system_model}

We consider a tandem system composed by two consecutive queues, where the first one is $M/M/1$. Packets are generated at the first system by a Poisson process with rate $\lambda$ and enter the first queue, whose service time is exponentially distributed with rate $\mu_1$. When a packet exits the first system, it enters the second one, whose service time is a constant $D$ or an exponential random variable with rate $\mu_2$. \edit{A simple diagram of the tandem is shown in Fig.~\ref{fig:edge}.} Both queues are of infinite size and are oblivious to the content of the packets: there is no preemption of the updates, i.e. an older packet is not removed from the queue when a new update comes from the same source. As explained in the introduction, we assume that the service times in the two systems are independent. The assumption is realistic for edge computing systems, as the communication and processing are usually independent, but not always verified in relay networks; it therefore needs a careful examination. Even if the service times are independent, however, the waiting times are not, as the queue at the second system depends on the output of the first one. In the following, we use the compact notation $p_{X|Y}(x|y)$ for the conditioned probability $p[X=x|Y=y]$. \glspl{pdf} are denoted by a lower-case $p$, and \glspl{cdf} by an upper-case $P$.

In a tandem queue, the packet generation times correspond to the arrival times at the first queue, whereas the receiving instants are the departure times in the second queue. We define the total system time for packet $i$ as $T_i$: when a packet is received, the \gls{aoi} is equal to %the system time of the packet
 $T_i$, i.e., the difference between the time $r_i$ when it is received by the destination and the time $g_i$ when it was generated. The \gls{paoi} (see Fig.~\ref{fig:aoi_triangles}) is the maximum value of the \gls{aoi}, i.e., the age at the instant immediately before the arrival of a new update. If we denote the interarrival time $Y_i=g_i-g_{i-1}$, then \gls{paoi} is given by
 \begin{equation}
     \Delta_i=r_i-g_{i-1}=r_i-g_i+g_i-g_{i-1}=T_i+Y_i
 \end{equation}
 
If packet $i$ arrives right after packet $i-1$, it will probably have to wait in the queue for it to depart the system: the system time $T_i$ depends on the interarrival time $Y_i$. The \gls{pdf} of the \gls{paoi} with value $\tau_i$, denoted as $p_{\Delta_i}(\tau_i)$, can then be computed by using the conditional system time probability $p_{T_i|Y_i}(t_i|y_i)$:
\begin{equation}
    p_{\Delta_i}(\tau_i)=\int_0^{\tau_i} p_{Y_i}(y_i)p_{T_i|Y_i}(\tau_i-y_i|y_i) dy_i
\end{equation}
We then need to compute %the conditioned system time probability
$p_{T_i|Y_i}(t_i|y_i)$.
For each system $j=1,2$ in the tandem, the system time $T_{i,j}$ is defined as the sum of the waiting time $W_{i,j}$ and the service time $S_{i,j}$. We also define $Y_{i,j}$, the interarrival time at system $j$. For $j=1$, we have $Y_{i,1}=Y_i$, while for $j=2$ :
\begin{equation}
 Y_{i,2}=g_i+T_{i,1} -(g_{i-1}+T_{i-1,1})=Y_i+T_{i,1}-T_{i-1,1}.
\end{equation}
Since the first queue is $M/M/1$, the system times for the two queues are independent, as proven by Reich~\cite{reich1963note} using Burke's theorem~\cite{burke1956output} and considering each system in steady state for packet $i-1$. If we consider system $j$ in steady state, i.e., we do not condition on $Y_{i-1,j}$, the system time $T_{i-1,j}$ is exponentially distributed with rate $\alpha_j=\mu_j-\lambda$.
However, the values of $Y_i$ and $T_i$ are correlated, and the computation of the \gls{paoi} needs to account for this fact. In the following, we give the conditional \gls{pdf} of the components of the \gls{paoi}, which will then be joined in the derivation.
6
First, we define the \emph{extended waiting time} $\Omega_{i,j}$ as the difference between the previous packet's system time and the interarrival time at the system, i.e., $\Omega_{i,j}=T_{i-1,j}-Y_{i,j}$.
The reason we named $\Omega_{i,j}$ the extended waiting time is that $W_{i,j}=[\Omega_{i,j}]^+$, where $[x]^+$ is equal to $x$ if it is positive and $0$ if $x$ is negative. From the definition of $\Omega_{i,j}$, we have:
\begin{equation}
\begin{aligned}
Y_{i,2}&=Y_i+(W_{i,1}+S_{i,1})-T_{i-1,1}=S_{i,1}+W_{i,1}-\Omega_{i,1}=S_{i,1}+[-\Omega_{i,1}]^+.
\end{aligned}
\end{equation}
since $W_{i,1}-\Omega_{i,1}=[\Omega_{i,1}]^+-\Omega_{i,1}=[-\Omega_{i,1}]^+$. In the following paragraphs, we derive the \gls{pdf} of the extended waiting time for the two system types that we are analyzing. Fig.~\ref{fig:packet} shows a possible realization of a packet's path through the tandem queue, highlighting the meaning of the extended waiting time: in the first system, in which packet $i$ is queued, it corresponds to the waiting time, while in the second, in which the packet is not queued and enters service immediately, its negative value corresponds to the time between the departure of packet $i-1$ from the second system and the arrival of packet $i$ at the same system. When $W=0$, we have a negative extended waiting time, as packet $i$ arrives after packet $i-1$ leaves the system. In general, we have $\Omega_{i,2}=T_{i-1,2}-Y_{i,2}$, and the system time for packet $i-1$ is $T_{i-1,2}=W_{i-1,2}+D$, while we know the interarrival time $Y_{i,2}=S_{i,1}+[-\Omega_{i,1}]^+$.

\begin{figure}[!t]
 \centering
 \resizebox{0.9\columnwidth}{!}{
\tikzset{  
block/.style    = {draw, rectangle, minimum height = 2cm, minimum width = 1em}}
\begin{tikzpicture}[auto]
  \node[rectangle,draw,dashed,inner sep=0pt,minimum width=2cm,minimum height=7cm,name=q1] at (0,0) {};
  \node[rectangle,dashed,right of=q1,node distance=5cm,draw,inner sep=0pt,minimum width=2cm,minimum height=7cm,name=s1] {};
  \node[rectangle,dashed,right of=s1,node distance=5cm,draw,inner sep=0pt,minimum width=2cm,minimum height=7cm,name=q2] {};
  \node[rectangle,dashed,right of=q2,node distance=5cm,draw,inner sep=0pt,minimum width=2cm,minimum height=7cm,name=s2] {};
  
    \node [name=label1,anchor=east] at (-1.5,3) {$g_{i-1}$};
    \node [name=label2,anchor=east] at (-1.5,1) {$g_i$};
    \node [name=q1l,anchor=south] at (0,3.75) {\small{Transmission buffer (sensor)}};
    \node [name=s1l,anchor=south] at (5,3.75) {\small{Communication link (sensor-BS)}};
    \node [name=q2l,anchor=south] at (10,3.75) {\small{Edge node queue (BS)}};
    \node [name=q2l,anchor=south] at (15,3.75) {\small{Computing} (BS)};

    \draw[densely dotted] (label1.east) to (-1,3);
    \draw[dashed,|-|] (-3,3) to node[midway,right] {$Y_{i,1}$} (-3,1);
    \draw[densely dotted] (label2.east) to (-1,1);
    
    \draw[->] (-1,3) to (1,1.75);
    \draw[->] (-1,1) to (1,0);
    \draw[dashed,|-|] (1.25,3) to node[midway,right] {$\Omega_{i-1,1}$} (1.25,1.75);
    \draw[dashed,|-|] (1.25,1) to node[midway,right] {$\Omega_{i,1}$} (1.25,0);

    \draw[densely dotted] (1,1.75) to (4,1.75);
    \draw[densely dotted] (1,0) to (6,0);

    \draw[->] (4,1.75) to (6,0);
    \draw[->] (4,0) to (6,-2);
    
    \draw[dashed,|-|] (6.25,1.75) to node[midway,right] {$S_{i-1,1}$} (6.25,0);
    \draw[dashed,|-|] (6.25,0) to node[midway,right] {$S_{i,1}$} (6.25,-2);
    
    \draw[densely dotted] (6,0) to (9,0);
    \draw[densely dotted] (6,-2) to (9,-2);
    
    \draw[->] (9,0) to (11,-0.5);
    \draw[->] (9,-2) to (11,-2);

    \draw[dashed,|-|] (11.25,0) to node[midway,right] {$\Omega_{i-1,2}$} (11.25,-0.5);
    \draw[dashed,|-|] (11.25,-2) to node[midway,right] {$-\Omega_{i,2}$} (11.25,-1.5);
    
    \draw[densely dotted] (11,-0.5) to (14,-0.5);
    \draw[densely dotted] (11,-1.5) to (16,-1.5);
    \draw[densely dotted] (14,-2) to (11,-2);
    
    \draw[->] (14,-0.5) to (16,-1.5);
    \draw[->] (14,-2) to (16,-3);
    
    \draw[dashed,|-|] (17.5,-0.5) to node[midway,right] {$S_{i-1,1}$} (17.5,-1.5);
    \draw[dashed,|-|] (17.5,-3) to node[midway,right] {$S_{i,1}$} (17.5,-2);
    \node [name=label3,anchor=west] at (16.5,-3) {$r_i$};
    \node [name=label4,anchor=west] at (16.5,-1.5) {$r_{i-1}$};
    \draw[densely dotted] (label3.west) to (16,-3);
    \draw[densely dotted] (label4.west) to (16,-1.5);

    \draw[dashed,|-|] (16.25,3) to node[midway,right] {$\Delta_i$} (16.25,-3);
    \draw[->] (-4,3.5) to node[midway,left] {Time} (-4,-3.5);

  \end{tikzpicture}
}
\caption{Schematic of the four steps a packet goes through in a tandem queue, highlighting the components of the \gls{paoi}.\vspace{-0.3cm}}
\label{fig:packet}
\end{figure}
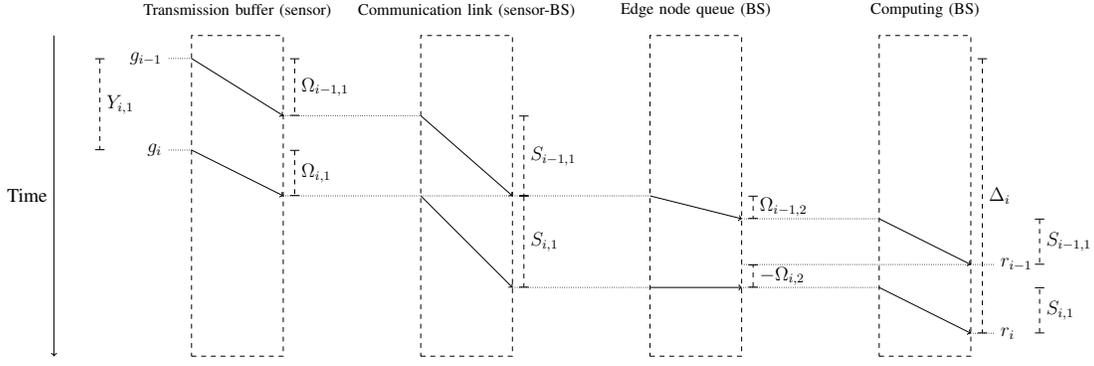

\begin{theorem}\label{th:md1waitcdf}\edit{
The \gls{cdf} of the waiting time for an $M/D/1$ queue with arrival rate $\lambda$ and service time $D$ is:
\begin{equation}
    P_W(w)=(1-\lambda D)\sum_{k=0}^{\left\lfloor\frac{w}{D}\right\rfloor} \frac{(-\lambda(w-kD))^k e^{\lambda(w-kD)}}{k!}.\label{eq:md1wait}
\end{equation}}
\end{theorem}

\begin{proof}\edit{
See the derivation by Erlang \cite{erlang1920telefon}.}
\end{proof}

\begin{corollary}\label{th:md1waitpdf}\edit{
We can find the queuing time \gls{pdf} by deriving the \gls{cdf} from~\eqref{eq:md1wait}:
\begin{equation}
p_W(w)=\begin{cases}(1-\lambda D)\left(\lambda e^{\lambda w}+\sum_{k=1}^{\left\lfloor\frac{w}{D}\right\rfloor}\left[ \frac{(-\lambda)^k(w-kD)^{k-1} e^{\lambda(w-kD)}}{k!}(k+\lambda(w-kD))\right]\right)&w>0;\\
         1-\lambda D &w=0.\label{eq:md1waitpdf}
       \end{cases}
\end{equation}}
\end{corollary}

The \gls{cdf} of the waiting time has a discontinuity, as the waiting time is exactly 0 with probability $1-\lambda D$, which corresponds to the probability of the packet finding an empty system and entering service immediately.

\begin{corollary}\label{th:md1total}\edit{
 As service time in an $M/D/1$ system is constant, we have the \gls{cdf} of the total time in the system:
\begin{equation}
  P_T(t)=P_W(t-D)u(t-D).
\end{equation}}
\end{corollary}
\begin{corollary}\label{th:md1omega}
\edit{The \gls{pdf} of $\Omega_{i,2}$ conditioned on $S_{i,1}$ and $\Omega_{i,1}$ is given by:
\begin{equation}
\begin{aligned}
 p_{\Omega_{i,2}|S_{i,1},\Omega_{i,1}}(\omega_{i,2}|s_{i,1},\omega_{i,1})&=p(W_{i-1,2}+D+S_{i,1}+[-\Omega_{i,1}]^+=\omega_{i,2}|S_{i,1}=s_{i,1},\Omega_{i,1}=\omega_{i,1})\\
 &=p_W(\omega_{i,2}+s_{i,1}+[-\omega_{i,1}]^+-D).
\end{aligned}
\end{equation}}
\end{corollary}

Fig.~\ref{fig:reception} shows this clearly: when $\Omega_{i,1}$ is negative, the extended queuing time $\Omega_{i,2}$ corresponds to $S_{i,1}+W_{i-1,2}+D-\Omega_{i,1}$, while when $\Omega_{i,1}$ is positive, packet $i$ starts service immediately after packet $i-1$ leaves the first system, and we have $\Omega_{i,2}=S_{i,1}+W_{i-1,2}+D$.
The waiting time in $M/D/1$ system is analytically derived, although it can give numeric problems for very large waiting times and high load~\cite{iversen1999waiting}. If the application requires the computation of very large waiting times with $\lambda D\simeq 1$, we suggest the use of more numerically stable methods from the relevant literature.

\begin{theorem}\label{th:mm1omega}
\edit{
 The \gls{pdf} of the extended waiting time in a $M/M/1$ - $M/M/1$ tandem queue is given by:
\begin{align}
    p_{\Omega_{i,j}|Y_{i,j}}\left(\omega_{i,j}|y_{i,j}\right)&=\alpha_j e^{-\alpha_j\left(\omega_{i,j}+y_{i,j}\right)}u(\omega_{i,j}+y_{i,j}),\label{eq:p_omega_y}
\end{align}
where $u(\cdot)$ is the step function.}
\end{theorem}

\begin{proof}
\edit{Knowing the \gls{pdf} of the system time $T_{i-1,j}$ from well-known results on $M/M/1$ queues, the interarrival time at the first relay $Y_{i,1}$ is exponentially distributed with rate $\lambda$, while in subsequent systems it is given by $Y_{i,2}=S_{i,1}+\left[-\Omega_{i,1}\right]^+$.}
\end{proof}

%\begin{theorem}\label{th:mm1omega}
% Knowing the \gls{pdf} of the system time $T_{i-1,j}$, we can derive the \gls{pdf} of the extended waiting time in case the second system is an $M/M/1$ queue, which will be useful in the next steps, in the same way we did for the $M/D/1$:
%\begin{align}
%    p_{\Omega_{i,j}|Y_{i,j}}\left(\omega_{i,j}|y_{i,j}\right)&=\alpha_j e^{-\alpha_j\left(\omega_{i,j}+y_{i,j}\right)}u(\omega_{i,j}+y_{i,j}),\label{eq:p_omega_y}
%\end{align}
%where $u(\cdot)$ is the step function. The interarrival time at the first relay $Y_{i,1}$ is exponentially distributed with rate $\lambda$, while in subsequent systems it is given by $Y_{i,2}=S_{i,1}+\left[-\Omega_{i,1}\right]^+$.
%\end{theorem}

\begin{corollary}\label{th:mm1uncond}
\edit{
We can combine~\eqref{eq:p_omega_y} with the definition of $Y_{i,2}$ to get:
\begin{equation}
\begin{aligned}
 p_{\Omega_{i,2}|S_{i,1},\Omega_{i,1}}(\omega_{i,2}|s_{i,1},\omega_{i,1})=\alpha_1 e^{-\alpha_1\left(\omega_{i,2}+s_{i,1}+[-\omega_{i,1}]^+\right)}u(\omega_{i,2}+s_{i,1}+[-\omega_{i,1}]^+).\label{eq:nextwait_mm1}
\end{aligned}
\end{equation}}
\end{corollary}

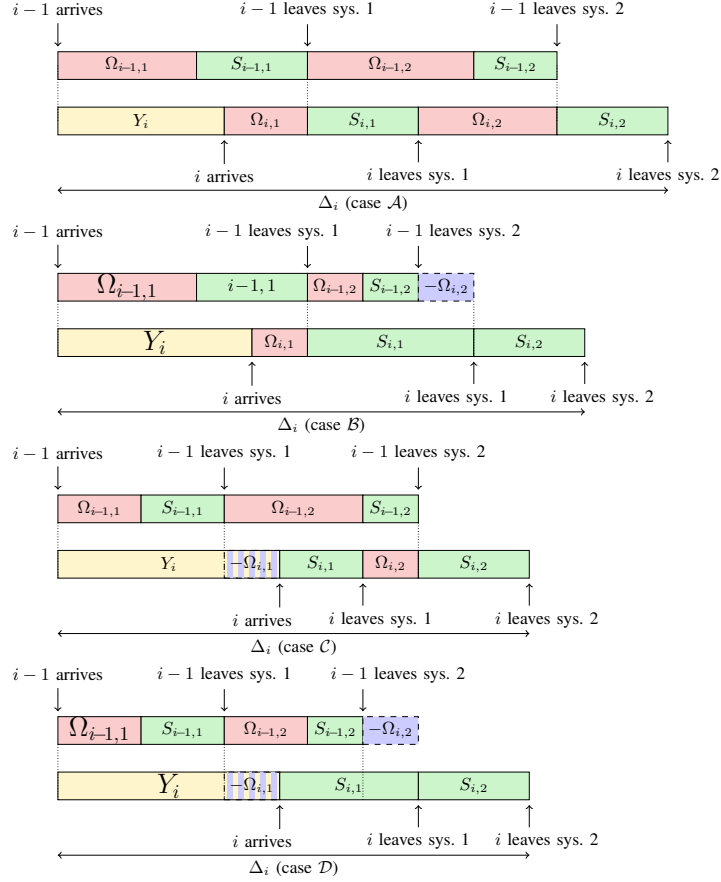
\begin{figure}[!t]
 \centering
 \resizebox{0.6\columnwidth}{!}{
 \footnotesize{
\tikzset{  
block/.style    = {draw, rectangle, minimum height = 2cm, minimum width = 1em}}
\begin{tikzpicture}[auto]

  \node[rectangle,draw,inner sep=0pt,minimum width=2.5cm,minimum height=0.5cm,fill={red!20},name=omone] at (0,0) {$\Omega_{i\!-\!1,1}$};
  \node[rectangle, right of=omone,node distance=2.25cm,draw,inner sep=0pt,minimum width=2cm,fill={green!20},minimum height=0.5cm,name=essone] {$S_{i\!-\!1,1}$};
  \node[rectangle, right of=essone,node distance=2.5cm,draw,inner sep=0pt,minimum width=3cm,fill={red!20},minimum height=0.5cm,name=omtwo] {$\Omega_{i\!-\!1,2}$};
  \node[rectangle, right of=omtwo,node distance=2.25cm,draw,inner sep=0pt,minimum width=1.5cm,fill={green!20},minimum height=0.5cm,name=esstwo] {$S_{i\!-\!1,2}$};
  \node[rectangle, draw,inner sep=0pt,minimum width=3cm,fill={yellow!20},minimum height=0.5cm,name=y] at (0.25, -1) {$Y_{i}$};
  \node[rectangle, right of=y,node distance=2.25cm,draw,inner sep=0pt,minimum width=1.5cm,fill={red!20},minimum height=0.5cm,name=omione] {$\Omega_{i,1}$};
  \node[rectangle, right of=omione,node distance=1.75cm,draw,inner sep=0pt,minimum width=2cm,fill={green!20},minimum height=0.5cm,name=essione] {$S_{i,1}$};
  \node[rectangle, right of=essione,node distance=2.25cm,draw,inner sep=0pt,minimum width=2.5cm,fill={red!20},minimum height=0.5cm,name=omitwo] {$\Omega_{i,2}$};
  \node[rectangle, right of=omitwo,node distance=2.25cm,draw,inner sep=0pt,minimum width=2cm,fill={green!20},minimum height=0.5cm,name=essitwo] {$S_{i,2}$};
  \node [name=label1] at (-1.25,1) {$i-1$ arrives};
  \node [name=label2] at (3.25,1) {$i-1$ leaves sys. 1};
  \node [name=label3] at (7.75,1) {$i-1$ leaves sys. 2};
  \node [name=label4] at (1.75,-2) {$i$ arrives};
  \node [name=label5] at (5.25,-2) {$i$ leaves sys. 1};
  \node [name=label6] at (9.75,-2) {$i$ leaves sys. 2};

  \draw[densely dotted] (3.25,0) to (3.25,-1);
  \draw[densely dotted] (7.75,0) to (7.75,-1);
  \draw[densely dotted] (-1.25,0) to (-1.25,-1);
  \draw[->] (label1.south) to (-1.25,0.35);
  \draw[->] (label2.south) to (3.25,0.35);
  \draw[->] (label3.south) to (7.75,0.35);
  \draw[->] (label4.north) to (1.75,-1.35);
  \draw[->] (label5.north) to (5.25,-1.35);
  \draw[->] (label6.north) to (9.75,-1.35);
  \draw[<->] (-1.25, -2.25) to node[midway,below] {$\Delta_i$ (case $\mathcal{A}$)} (9.75, -2.25);
  
  \node[rectangle,draw,inner sep=0pt,minimum width=2.5cm,minimum height=0.5cm,fill={red!20},name=omoneb] at (0,-4) {\large$\Omega_{i\!-\!1,1}$};
  \node[rectangle, right of=omoneb,node distance=2.25cm,draw,inner sep=0pt,minimum width=2cm,fill={green!20},minimum height=0.5cm,name=essoneb] {${i\!-\!1,1}$};
  \node[rectangle, right of=essoneb,node distance=1.5cm,draw,inner sep=0pt,minimum width=1cm,fill={red!20},minimum height=0.5cm,name=omtwob] {$\Omega_{i\!-\!1,2}$};
  \node[rectangle, right of=omtwob,node distance=1cm,draw,inner sep=0pt,minimum width=1cm,fill={green!20},minimum height=0.5cm,name=esstwob] {$S_{i\!-\!1,2}$};
  \node[rectangle, right of=esstwob,node distance=1cm,draw, dashed,inner sep=0pt,minimum width=1cm,fill={blue!20},minimum height=0.5cm,name=omitwob] {$-\Omega_{i,2}$};
  \node[rectangle, draw,inner sep=0pt,minimum width=3.5cm,fill={yellow!20},minimum height=0.5cm,name=yb] at (0.5, -5) {\large$Y_{i}$};
  \node[rectangle, right of=yb,node distance=2.25cm,draw,inner sep=0pt,minimum width=1cm,fill={red!20},minimum height=0.5cm,name=omioneb] {$\Omega_{i,1}$};
  \node[rectangle, right of=omioneb,node distance=2cm,draw,inner sep=0pt,minimum width=3cm,fill={green!20},minimum height=0.5cm,name=essioneb] {$S_{i,1}$};
  \node[rectangle, right of=essioneb,node distance=2.5cm,draw,inner sep=0pt,minimum width=2cm,fill={green!20},minimum height=0.5cm,name=essitwob] {$S_{i,2}$};
  \node [name=label7] at (-1.25,-3) {$i-1$ arrives};
  \node [name=label8] at (3.25,-3) {$i-1$ leaves sys. 1$\quad\quad\quad\quad$};
  \node [name=label9] at (5.25,-3) {$\quad\quad\quad\quad i-1$ leaves sys. 2};
  \node [name=label10] at (2.25,-6) {$i$ arrives};
  \node [name=label11] at (6.25,-6) {$i$ leaves sys. 1$\quad\quad$};
  \node [name=label12] at (8.25,-6) {$\quad\quad i$ leaves sys. 2};
  \draw[densely dotted] (3.25, -5.25) to (3.25, -4.25);
  \draw[densely dotted] (-1.25, -5.25) to (-1.25, -4.25);
  \draw[densely dotted] (6.25, -5.25) to (6.25, -4.25);
  \draw[->] (label7.south) to (-1.25,-3.65);
  \draw[->] (label8.south) to (3.25,-3.65);
  \draw[->] (label9.south) to (5.25,-3.65);
  \draw[->] (label10.north) to (2.25,-5.35);
  \draw[->] (label11.north) to (6.25,-5.35);
  \draw[->] (label12.north) to (8.25,-5.35);
  \draw[<->] (-1.25, -6.25) to node[midway,below] {$\Delta_i$ (case $\mathcal{B}$)} (8.25, -6.25);
  
  \node[rectangle,draw,inner sep=0pt,minimum width=1.5cm,minimum height=0.5cm,fill={red!20},name=omonec] at (-0.5,-8) {$\Omega_{i\!-\!1,1}$};
  \node[rectangle, right of=omonec,node distance=1.5cm,draw,inner sep=0pt,minimum width=1.5cm,fill={green!20},minimum height=0.5cm,name=essonec] {$S_{i\!-\!1,1}$};
  \node[rectangle, right of=essonec,node distance=2cm,draw,inner sep=0pt,minimum width=2.5cm,fill={red!20},minimum height=0.5cm,name=omtwoc] {$\Omega_{i\!-\!1,2}$};
  \node[rectangle, right of=omtwoc,node distance=1.75cm,draw,inner sep=0pt,minimum width=1cm,fill={green!20},minimum height=0.5cm,name=esstwoc] {$S_{i\!-\!1,2}$};
  \node[rectangle, draw,inner sep=0pt,minimum width=4cm,fill={yellow!20},minimum height=0.5cm,name=yc] at (0.75, -9) {$Y_{i}$};
  \draw[draw=none,fill={blue!20}] (1.8,-9.24) rectangle (1.9,-8.76);
  \draw[draw=none,fill={blue!20}] (2,-9.24) rectangle (2.1,-8.76);
  \draw[draw=none,fill={blue!20}] (2.2,-9.24) rectangle (2.3,-8.76);
  \draw[draw=none,fill={blue!20}] (2.4,-9.24) rectangle (2.5,-8.76);
  \draw[draw=none,fill={blue!20}] (2.6,-9.24) rectangle (2.7,-8.76);
   \node[rectangle, right of=yc,node distance=1.5cm,draw, dashed,inner sep=0pt,minimum width=1cm,minimum height=0.5cm,name=omionec] {$-\Omega_{i,1}$};
  \node[rectangle, right of=omionec,node distance=1.25cm,draw,inner sep=0pt,minimum width=1.5cm,fill={green!20},minimum height=0.5cm,name=essionec] {$S_{i,1}$};
  \node[rectangle, right of=essionec,node distance=1.25cm,draw,inner sep=0pt,minimum width=1cm,fill={red!20},minimum height=0.5cm,name=omitwoc] {$\Omega_{i,2}$};
  \node[rectangle, right of=omitwoc,node distance=1.5cm,draw,inner sep=0pt,minimum width=2cm,fill={green!20},minimum height=0.5cm,name=essitwoc] {$S_{i,2}$};
  \node [name=label7] at (-1.25,-7) {$i-1$ arrives};
  \node [name=label8] at (1.75,-7) {$i-1$ leaves sys. 1};
  \node [name=label9] at (5.25,-7) {$i-1$ leaves sys. 2};
  \node [name=label10] at (2.75,-10) {$i$ arrives$\quad\quad$};
  \node [name=label11] at (4.25,-10) {$\quad\quad i$ leaves sys. 1};
  \node [name=label12] at (7.25,-10) {$\quad\quad i$ leaves sys. 2};
  \draw[densely dotted] (1.75, -9.25) to (1.75, -8.25);
  \draw[densely dotted] (-1.25, -9.25) to (-1.25, -8.25);
  \draw[densely dotted] (5.25, -9.25) to (5.25, -8.25);
  \draw[->] (label7.south) to (-1.25,-7.65);
  \draw[->] (label8.south) to (1.75,-7.65);
  \draw[->] (label9.south) to (5.25,-7.65);
  \draw[->] (label10.north) to (2.75,-9.35);
  \draw[->] (label11.north) to (4.25,-9.35);
  \draw[->] (label12.north) to (7.25,-9.35);
  \draw[<->] (-1.25, -10.25) to node[midway,below] {$\Delta_i$ (case $\mathcal{C}$)} (7.25, -10.25);
  
    \node[rectangle,draw,inner sep=0pt,minimum width=1.5cm,minimum height=0.5cm,fill={red!20},name=omoned] at (-0.5,-12) {\large$\Omega_{i\!-\!1,1}$};
  \node[rectangle, right of=omoned,node distance=1.5cm,draw,inner sep=0pt,minimum width=1.5cm,fill={green!20},minimum height=0.5cm,name=essoned] {$S_{i\!-\!1,1}$};
  \node[rectangle, right of=essoned,node distance=1.5cm,draw,inner sep=0pt,minimum width=1.5cm,fill={red!20},minimum height=0.5cm,name=omtwod] {$\Omega_{i\!-\!1,2}$};
  \node[rectangle, right of=omtwod,node distance=1.25cm,draw,inner sep=0pt,minimum width=1cm,fill={green!20},minimum height=0.5cm,name=esstwod] {$S_{i\!-\!1,2}$};
  \node[rectangle, right of=esstwod,node distance=1cm,draw, dashed,inner sep=0pt,minimum width=1cm,fill={blue!20},minimum height=0.5cm,name=omitwod] {$-\Omega_{i,2}$};
  \node[rectangle, draw,inner sep=0pt,minimum width=4cm,fill={yellow!20},minimum height=0.5cm,name=yd] at (0.75, -13) {\large$Y_{i}$};
  \draw[draw=none,fill={blue!20}] (1.8,-13.24) rectangle (1.9,-12.76);
  \draw[draw=none,fill={blue!20}] (2,-13.24) rectangle (2.1,-12.76);
  \draw[draw=none,fill={blue!20}] (2.2,-13.24) rectangle (2.3,-12.76);
  \draw[draw=none,fill={blue!20}] (2.4,-13.24) rectangle (2.5,-12.76);
  \draw[draw=none,fill={blue!20}] (2.6,-13.24) rectangle (2.7,-12.76);
   \node[rectangle, right of=yd,node distance=1.5cm,draw, dashed,inner sep=0pt,minimum width=1cm,minimum height=0.5cm,name=omioned] {$-\Omega_{i,1}$};
  \node[rectangle, right of=omioned,node distance=1.75cm,draw,inner sep=0pt,minimum width=2.5cm,fill={green!20},minimum height=0.5cm,name=essioned] {$S_{i,1}$};
  \node[rectangle, right of=essioned,node distance=2.25cm,draw,inner sep=0pt,minimum width=2cm,fill={green!20},minimum height=0.5cm,name=essitwod] {$S_{i,2}$};
  \node [name=label7] at (-1.25,-11) {$i-1$ arrives};
  \node [name=label8] at (1.75,-11) {$i-1$ leaves sys. 1};
  \node [name=label9] at (4.25,-11) {$\quad\quad\quad\quad i-1$ leaves sys. 2};
  \node [name=label10] at (2.75,-14) {$i$ arrives$\quad\quad$};
  \node [name=label11] at (5.25,-14) {$i$ leaves sys. 1};
  \node [name=label12] at (7.25,-14) {$\quad\quad i$ leaves sys. 2};
  \draw[densely dotted] (1.75, -13.25) to (1.75, -12.25);
  \draw[densely dotted] (-1.25, -13.25) to (-1.25, -12.25);
  \draw[densely dotted] (4.25, -13.25) to (4.25, -12.25);
  \draw[->] (label7.south) to (-1.25,-11.65);
  \draw[->] (label8.south) to (1.75,-11.65);
  \draw[->] (label9.south) to (4.25,-11.65);
  \draw[->] (label10.north) to (2.75,-13.35);
  \draw[->] (label11.north) to (5.25,-13.35);
  \draw[->] (label12.north) to (7.25,-13.35);
  \draw[<->] (-1.25, -14.25) to node[midway,below] {$\Delta_i$ (case $\mathcal{D}$)} (7.25, -14.25);
  
\end{tikzpicture}
}}
\caption{Schematic of the components of the \gls{paoi}. (case $\mathcal{A}$): packet $i$ has to wait in both queues. (case $\mathcal{B}$): packet $i$ waits only in the first queue. (case $\mathcal{C}$): packet $i$ waits only in the second queue. (case $\mathcal{D}$): packet $i$ is immediately served in both systems. Observe the negative expected waiting time $\Omega_{i,j}$ in the cases with empty queue(s).\vspace{-0.3cm}}
\label{fig:reception}
\end{figure}

To compute the exact \gls{pdf} of \gls{paoi} in the 2-system case ($j \in {1, 2}$), we distinguish between free and busy systems at each node, i.e., we condition the \gls{pdf} on the state of each system when packet $i$ arrives to it and calculate it separately for the four possible combinations. Case $\mathcal{A}$ is defined as $\Omega_{i,1}>0\wedge\Omega_{i,2}>0$, while in case $\mathcal{B}$ we have $\Omega_{i,1}>0\wedge\Omega_{i,2}\leq0$. Similarly, in case $\mathcal{C}$ we have $\Omega_{i,1}\leq0\wedge\Omega_{i,2}>0$ and in case $\mathcal{D}$ we get $\Omega_{i,1}\leq0\wedge\Omega_{i,2}\leq0$. An example of the relevant values in the four cases are shown in Fig.~\ref{fig:reception}: in case $\mathcal{A}$, packet $i$ is queued in both systems, as the previous packet is still in the system when $i$ arrives in each. The two extended queuing times (shown in red) are positive. In case $\mathcal{B}$, packet $i-1$ has already left the second system when packet $i$ leaves the first: the extended queuing time (shown in blue with a dashed outline) is negative, and packet $i$ enters service in the second system as soon as it arrives. In case $\mathcal{C}$, it is the first system that is empty when packet $i$ arrives, and in case $\mathcal{D}$, both systems are empty, and the packet enters service directly at both.

\edit{We can then exploit Corollaries~\ref{th:md1waitpdf}-\ref{th:md1omega} to compute the conditioned \gls{pdf} of the total time at the second system. The \gls{paoi} can then simply be computed by unconditioning over the values of $S_{i,1}$, $\Omega_{i,1}$, and $Y_i$, simply applying the law of total probability until the \gls{pdf} for the given case is derived. The division in 4 cases is not strictly necessary, but it reduces the number of terms in the equations considerably with respect to deriving the \gls{pdf} for the general case directly. The computation for the $M/M/1$ -- $M/M/1$ tandem follows the same reasoning, using the results in Theorem~\ref{th:mm1omega} and its Corollary~\ref{th:mm1uncond}.}

\begin{definition}\label{th:totalpdf}\edit{The \gls{pdf} of the \gls{paoi} is:
\begin{equation}
\begin{aligned}
  p_{\Delta}(\tau)=p_{\Delta|\mathcal{A}}\left(\tau\right)p(\mathcal{A})+p_{\Delta|\mathcal{B}}\left(\tau\right)p(\mathcal{B})+p_{\Delta|\mathcal{C}}\left(\tau\right)p(\mathcal{C})+p_{\Delta|\mathcal{D}}\left(\tau\right)p(\mathcal{D}).
\end{aligned} \label{eq:paoi}
\end{equation}}
\end{definition}
where $p_{\Delta|\mathcal{X}}$ is the \gls{pdf} of the \gls{paoi} in case $X$ and $p(\mathcal{X})$ is the probability of case $X$ happening. The definition comes from the application of the total law of probability.

In case $\mathcal{A}$, packet $i$ is queued at both systems, and the packet will have the highest queuing delay and system time. The case with the lowest system time is case $\mathcal{D}$, in which the packet experiences no queuing. However, these intuitive relations do not necessarily hold for the \gls{paoi}, as the interarrival time between update packets can play a major role. In the analysis of the four cases, we will omit the packet index $i$ wherever possible for the sake of readability. For a quick overview of the notation used in the rest of this paper, we refer the reader to Table \ref{tab:notation}.

\begin{table}[t]\centering
    \footnotesize
	\caption{Main notation used in the paper.}
	\begin{tabular}[c]{cl}
		\toprule
		Symbol & Description \\
		\midrule
        $\lambda$ & Packet generation rate\\
        $\mu_j$ & Service rate of $M/M/1$ system $j$\\
        $D$ & Service time of the $M/D/1$ system\\
        $\alpha_j=\mu_j-\lambda$ & Response rate of system $j$\\
        $S_{i,j}$ & Service time in system $j$ for packet $i$\\
        $Y_{i,j}$ & Interarrival time in system $j$ for packet $i$\\
        $\Omega_{i,j}=T_{i-1,j}-Y_{i,j}$ & Extended waiting time in $j$ for packet $i$\\
        $W_{i,j}=[\Omega_{i,j}]^+$ & Waiting time in system $j$ for packet $i$\\
        $T_{i,j}=S_{i,j}+W_{i,j}$ & Total time in system $j$ for packet $i$\\
        $\Delta_i=Y_i+T_i$ & \gls{paoi} for packet $i$\\
		\bottomrule
	\end{tabular}
	\label{tab:notation}
\end{table}

\section{PAoI distribution for the $M/M/1$ -- $M/D/1$ tandem}\label{sec:md1}

First, we analyze the tandem of an $M/M/1$ and an $M/D/1$ system, using the \gls{pdf} of the waiting time from \eqref{eq:md1waitpdf}.
\begin{definition}\label{th:theta}\edit{
The auxiliary function $\theta(M,\beta)$, which we will use in the later derivations to make the notation more compact, is defined as:
\begin{equation}
\begin{aligned}
\theta(M,\beta)=&\int_0^M p_W(w)e^{\beta w}dw \,\forall\beta\neq0,\beta\neq-\lambda\\
=&\sum_{k=1}^{\left\lfloor\frac{M}{D}\right\rfloor}(1-\lambda D) \int\displaylimits_{kD}^M\frac{(-\lambda)^k(w-kD))^{k-1} e^{\lambda(w-kD)+\beta w}}{k!}(k+\lambda(w-kD)) dw\\
&+(1-\lambda D)\int\displaylimits_0^M\lambda e^{(\lambda+\beta) w}dw\\
=&(1-\lambda D)\Bigg(\frac{\lambda(e^{(\lambda+\beta) M}-1)}{\lambda+\beta}+\sum_{k=1}^{\left\lfloor\frac{M}{D}\right\rfloor}\Bigg[e^{\beta k D}\bigg(\frac{\beta\lambda^k}{(\lambda+\beta)^{k+1}}-e^{(\lambda+\beta)(M-kD)}\\
&\times\bigg(\frac{(-\lambda)^{k+1}(M-kD)^k}{(\lambda+\beta)k!}+\sum_{j=0}^{k-1}\frac{\lambda^k\beta(kD-M)^j}{(\lambda+\beta)^{k-j+1}j!} \bigg)\bigg)\Bigg]\Bigg).\\
\end{aligned} \label{eq:auxiliary}
\end{equation}
If $\beta=0$, the result of the integral is simply the waiting time \gls{cdf}, i.e., $\theta(M,0)=P_W(M)$.
If $\beta=-\lambda$, we have:
\begin{equation}
\begin{aligned}
\theta(M,-\lambda)=&\int_0^M p_W(w)e^{\beta w}dw\\
=&(1-\lambda D)\left(\lambda M+\sum_{k=1}^{\left\lfloor\frac{M}{D}\right\rfloor} \int\displaylimits_{kD}^M\frac{(-\lambda)^k(w-kD))^{k-1} e^{-\lambda kD}}{k!}(k+\lambda(w-kD)) dw\right)\\
=&(1-\lambda D)\left(\lambda M+e^{\beta kD}(-\lambda)^k\left(\frac{(M-kD)^k}{k!}+\frac{\lambda(M-kD)^{k+1}}{(k+1)!}\right)\right).
\end{aligned}
\end{equation}}
\end{definition}

If we consider a simple $M/D/1$ queue (i.e., not in tandem), it is easy to derive the distribution of the \gls{paoi}, as we have
\begin{equation}
  \Delta_i=D+\max(Y_i,T_{i-1}).
\end{equation}
Therefore, the \gls{paoi} is lower than $\tau$ when both $Y_i$ and $T_{i-1}$ are smaller than $\tau-D$, and we can write the \gls{cdf} as: 
%derive the \gls{cdf} of the \gls{paoi} for the single $M/D/1$ queue, as the \gls{paoi} is lower than $\tau$ when both $Y_i$ and $T_{i-1}$ are smaller than $\tau-D$:
\begin{equation}
  \begin{aligned}
    P_{\Delta}(\tau)=\int_0^{\tau-D}P_W(\tau-2D)\lambda e^{-\lambda y} dy  =(1-e^{\lambda(T-\tau)})P_W(\tau-2D)u(\tau-2D).
  \end{aligned}
\end{equation}
Things are more complex in the tandem system, in which an $M/M/1$ queue feeds the $M/D/1$ queue. Thanks to Burke's theorem~\cite{burke1956output}, we can consider both systems to be in steady state for packet $i-1$, distinguishing the same four cases $\mathcal{A}$-$\mathcal{D}$ described in Section~\ref{sec:system_model} and Figure~\ref{fig:reception}. 

\subsection{The packet is queued at both systems}

We start by considering the conditional \gls{cdf} of the \gls{paoi} in case $\mathcal{A}$, in which the $i$-th packet is queued at both systems (i.e., $\Omega_{i,1}>0\wedge\Omega_{i,2}>0$). The probability of a packet being in case $\mathcal{A}$ is given by two concurrent events: first, packet $i$ must find the first system busy, which is equivalent to stating that $T_{i-1,1}>Y_i$. Then, the second system must also be busy when the packet arrives to it, so we have $T_{i-1,1}+W_{i-1,2}+S_{i-1,2}>Y_i+S_{i,1}$:
\begin{equation}
\begin{aligned}
 p(\mathcal{A})&=p(\Omega_1>0)p(\Omega_2>0|\Omega_1>0)\\
  &=p(Y_i<T_{i-1,1},S_{i,1}< T_{i-1,1}-Y_i+W_{i-1,2}+D)\\
  &=\int\displaylimits_0^\infty \int\displaylimits_0^{t_1} p_{y_1}(y_1) p_{T_1}(t_1) dy_1 dt_1\int\displaylimits_0^\infty P_{S_1}(w+D) p_W(w) dw \\
 &=\rho_1\left(1-(1-\lambda D)e^{-\mu_1 D}-e^{-\mu_1 D}\lim_{M\rightarrow\infty}\theta(M,-\mu_1)\right) \\
 &=\rho_1 \left(1-(1-\lambda D)e^{-\mu_1 D}\left(1+\frac{\lambda(e^{\mu_1 D}-1)}{\alpha_1 e^{\mu_1 D}+\lambda)}\right)\right).
 \end{aligned}
\end{equation}
The conditioned distribution of the \gls{paoi} in case $\mathcal{A}$ is:
\begin{equation}
\begin{aligned}
    p_{\Delta_i|T_{i-1,1},\mathcal{A}}(\tau|t_1)=&\frac{p_W(\tau-t_1-2D)P_{Y_i+S_i}(\tau-D)P_{Y_i}(t_1)}{p(\mathcal{A})}
    \\
    =&\frac{p_W(\tau-t_1-2D) (1-e^{-\lambda t_1})(1-e^{-\mu_1(\tau-t_1-D)}) u(\tau-t_1-2D)}{p(\mathcal{A})}.
\end{aligned}
\end{equation}

We need to consider the case for $w=0$ separately, as there is a discontinuity in the \gls{cdf}. We can now uncondition the distribution by substituting $p_W(w)$ and applying the law of total probability:
\begin{equation}
\begin{aligned}
  p_{\Delta|\mathcal{A}}(\tau)=&\int\displaylimits_0^{\mathclap{\tau-2D}}p_{\Delta_i|T_{i-1,1},\mathcal{A}}(\tau|t_1)\alpha_1 e^{-\alpha_1t_1} dt_1+ \alpha_1(1-\lambda D)e^{-\alpha_1(\tau-2D)}(1-e^{-\lambda(\tau-2D)}-e^{-\mu_1D}) \\
  =&\int\displaylimits_0^{\mathclap{\tau-2D}}\frac{\alpha_1}{p(\mathcal{A})}\left(e^{\alpha_1(w-\tau+2D)}-e^{\mu_1(w-\tau+2D)}-e^{-\mu_1 D-\alpha_1(\tau-2D)-\lambda w}+e^{-\mu_1(\tau-D)}\right)p_W(w)dw\\
  &+ \frac{\alpha_1(1-\lambda D)}{p(\mathcal{A})}(1-e^{-\lambda t})(1-e^{-\mu_1(\tau-t-D)})e^{-\alpha_1(\tau-2D)}.\label{eq:casea_md1}
  \end{aligned}
\end{equation}

We can then substitute the auxiliary function defined in Lemma~\ref{th:theta} in~\eqref{eq:casea_md1}, getting the \gls{pdf} of the \gls{paoi} in case $\mathcal{A}$:
\begin{equation}
   \begin{aligned}
  p_{\Delta|\mathcal{A}}(\tau)=&\frac{\alpha_1}{p(\mathcal{A})}\big(e^{-\alpha_1(\tau+2D)}\theta(\tau-2D,\alpha_1)-e^{-\mu_1(\tau-2D)}\theta(\tau-2D,\mu_1)\\
  &-e^{-\mu_1 D-\alpha_1(\tau-2D)}\theta(\tau-2D,-\lambda)+e^{-\mu_1(\tau-D)}(P_W(\tau-2D)-(1-\lambda D)\big).\label{eq:casea_md1_complete}
  \end{aligned} 
\end{equation}

\subsection{The packet is only queued at the first system}

We now consider case $\mathcal{B}$, in which there is queuing at the first system, but not at the second. This is equivalent to stating that $T_{i-1,1}>Y_{i,1}\wedge\Omega_{i-1,2}+D\leq S_{i,1}$. Consequently, we get the following probability for case $\mathcal{B}$:
\begin{equation}
\begin{aligned}
 p(\mathcal{B})&=p(\Omega_1>0)p(\Omega_2\leq0|\Omega_1>0)\\
 &=p(Y_i<T_{i-1,1},S_{i,1}\geq T_{i-1,1}-Y_i+W_{i-1,2}+D)\\
 &=\int\displaylimits_0^\infty \int\displaylimits_0^{t_1} p_{y_1}(y_1) p_{T_1}(t_1) dy_1 dt_1\int\displaylimits_0^\infty (1-P_{S_1}(w+D)) p_W(w) dw \\
 &=\rho_1e^{-\mu_1 D}\left((1-\lambda D)+\lim_{M\rightarrow\infty}\theta(M,-\mu_1)\right)\\
 &=(1-\lambda D)\rho_1e^{-\mu_1 D}\left(1+\frac{\lambda(e^{\mu_1 D}-1)}{\alpha_1 e^{\mu_1 D}+\lambda}\right).
 \end{aligned}
\end{equation}
The \gls{paoi} in this case is given by $T_{i-1,1}+S_{i,1}+D$. Its conditional \gls{pdf} in this case is given by:
\begin{equation}
\begin{aligned}
    p_{\Delta|T_{i-1,1},\mathcal{B}}(\tau|t_1)=&\frac{ p_{S_1}(\tau-t_1-D)P_{Y_1}(t_1)P_W(\tau-t_1-2D)}{p(\mathcal{B})}\\
    =&\frac{\mu_1 e^{-\mu_1(\tau-t_1-D)}P_W(\tau-t_1-2D)(1-e^{-\lambda t_1})}{p(\mathcal{B})}.
\end{aligned}
\end{equation}
We can now uncondition this probability to obtain the \gls{pdf} of the \gls{paoi} in case $\mathcal{B}$:
\begin{equation}
\begin{aligned}
p_{\Delta|\mathcal{B}}(\tau)=&\int_D^{\tau-D}\frac{P_W(s_1-D)}{p(\mathcal{B})}\alpha_1 e^{-\alpha_1(\tau-s_1-D)}\mu_1 e^{-\mu_1 s_1}(1-e^{-\lambda(\tau-s_1-D)}) dx\\
=&\alpha_1\frac{\mu_1}{p(\mathcal{B})} e^{-\alpha_1(\tau-D)-\lambda D}\left(\int_0^{\tau-2D}P_W(w)e^{-\lambda w}dw -\int_0^{\tau-2D}P_W(w) dw\right)\\
=&\frac{\alpha_1(1-\lambda D)}{\rho_1 p(\mathcal{B})}e^{-\alpha_1(\tau-D)-\lambda D}\sum_{k=0}^{\left\lfloor\frac{\tau-2D}{D}\right\rfloor}\left(e^{-\lambda(\tau-2D)}-\sum_{j=0}^{k+1}\frac{(-\lambda)^j(\tau-(k+2)D)^j e^{-\lambda kD}}{j!}\right).
\end{aligned}
\end{equation}

\subsection{The packet is only queued at the second system}

We then consider case $\mathcal{C}$, in which there is no queuing at the first system, but the packet is queued at the second. This is equivalent to stating that $T_{i-1,1}\leq Y_{i,1}\wedge\Omega_{i-1,2}+D>S_{i,1}-\Omega_{i,1}$. Consequently, we get the following probability for case $\mathcal{C}$:
\begin{equation}
\begin{aligned}
 p(\mathcal{C})&=p(\Omega_1\leq0,\Omega_2>0)\\
   &=p(Y_i\geq T_{i-1,1},S_{i,1}<T_{i-1,1}-Y_i+W_{i-1,2}+D)\\
 &=\int\displaylimits_0^\infty \int\displaylimits_0^{y_1}\int\displaylimits_{\max(0,y_1-t_1-D)}^\infty P_{S_1}(w+t_1+D-y_1) p_W(w) p_{y_1}(y_1) p_{T_1}(t_1) dwdt_1 dy_1 \\
  &=\lambda D - p(\mathcal{A}).
 \end{aligned}
\end{equation}
The \gls{paoi} in this case is given by $T_{i-1,1}+\Omega_{i-1,2}+2D$. Since we know that we are in case $\mathcal{C}$, we have:
\begin{equation}
\begin{aligned}
    p_{\Delta|Y_{i,1},\Omega_{i-1,2},\mathcal{C}}(\tau|y_1,w)=&\frac{P_{S_1}(\tau-y_1-D)P_{Y_1}(y_1)p_T(\tau-2D-w)u(y_1+w+2D-\tau)}{p(\mathcal{C})}\\
    =&\frac{\alpha_1(1-e^{-\mu_1(\tau-y_1-D)})e^{\alpha_1(w-\tau+2D)}u(y_1+w+2D-\tau)}{p(\mathcal{C})}.
\end{aligned}
\end{equation}
We can now uncondition this probability on $Y_{i,1}$ by applying the law of total probability:
\begin{equation}
\begin{aligned}
    p_{\Delta|\Omega_{i-1,2},\mathcal{C}}(\tau|w)=&\int\displaylimits_{\mathclap{\tau-2D-w}}^{\tau-D}\frac{(1- e^{-\mu_1(\tau-y_1-D)})\lambda e^{-\lambda y_1}\alpha_1e^{\alpha_1(w-\tau+2D)}}{p(\mathcal{C})}dy_1\\
    =&\frac{e^{\mu_1(D-\tau)}(\alpha_1e^{\mu_1(w+D)}-\mu_1 e^{\alpha_1(w+D)}+\lambda)}{p(\mathcal{C})}.
    \end{aligned}
\end{equation}
We can now uncondition again on $T_{i-1,i}$ and get the \gls{pdf} of the \gls{paoi} in case $\mathcal{D}$:
\begin{equation}
\begin{aligned}
p_{\Delta|\mathcal{C}}(\tau)=&\int\displaylimits_0^{\tau-2D}p_{\Delta|\Omega_{i-1,2},\mathcal{C}}(\tau|w)p_W(w) dw+\frac{1-\lambda D}{p(\mathcal{C})}p_{\Delta|\Omega_{i-1,2},\mathcal{C}}(\tau|0)\\
=&\frac{e^{-\mu_1(\tau-D)}}{p(\mathcal{C})}\left(\alpha_1e^{\mu_1D}\theta(\tau-2D,\mu_1)-\mu_1e^{\alpha_1D}\theta(\tau-2D,\alpha_1)+\lambda P_W(\tau-2D)\right)\\
&+\frac{(1-\lambda D)e^{-\mu_1(\tau-2D)}}{p(\mathcal{C})}\left(\alpha_1-\mu_1e^{-\lambda D}+\lambda e^{-\mu_1 D}\right).
\end{aligned}
\end{equation}

\subsection{The packet is not queued at either system}

Finally, we consider case $\mathcal{D}$, in which there is no queuing at either system. This is equivalent to stating that $T_{i-1,1}\leq Y_{i,1}\wedge\Omega_{i-1,2}+D\leq S_{i,1}-\Omega_{i,1}$. Consequently, we get the following probability for case $\mathcal{D}$:
\begin{equation}
\begin{aligned}
 p(\mathcal{D})&=p(\Omega_1\leq0,\Omega_2\leq0)\\
   &=p(Y_i\geq T_{i-1,1},S_{i,1}\geq T_{i-1,1}-Y_i+W_{i-1,2}+D)\\
 &=\int\displaylimits_0^\infty \int\displaylimits_0^{y_1}\int\displaylimits_{\max(y_1-t_1-D,0)}^\infty (1-P_{S_1}(w+t_1+D-y_1)) p_W(w) p_{y_1}(y_1) p_{T_1}(t_1) dwdt_1 dy_1 \\
  &=(1-\lambda D)-p(\mathcal{B}).
 \end{aligned}
\end{equation}
The \gls{paoi} in this case is given by $Y_{i,1}+S_{i,1}+D$. Since we know that we are in case $\mathcal{D}$, we can apply Bayes' theorem to get:
\begin{equation}
\begin{aligned}
    p_{\Delta|\mathcal{D},Y_{i,1},T_{i-1,1}}(\tau|y_1,t_1)=&\frac{p_{S_1}(\tau-y_1-D)P_{Y_1}(y_1)P_W(\tau-t_1-2D)u(y_1-t_1)}{p(\mathcal{D})}\\
    =&\frac{\mu_1 e^{-\mu_1(\tau-y_1-D)}P_W(\tau-t_1-2D)u(y_1-t_1)}{p(\mathcal{D})}.
\end{aligned}
\end{equation}
We can now uncondition this probability on $Y_{i,1}$ by applying the law of total probability:
\begin{equation}
\begin{aligned}
    p_{\Delta|\mathcal{D},T_{i-1,1}}(\tau|t_1)=&\int_t^{\tau-D}\frac{\mu_1 e^{-\mu_1(\tau-y_1-D)}P_W(\tau-t_1-2D)u(y_1-t_1)\lambda e^{-\lambda y_1}}{p(\mathcal{D})}dy_1\\
    =&\frac{\lambda\mu_1}{\alpha_1 p(\mathcal{D})}P_W(\tau-t_1-2D)e^{-\mu_1(\tau-D)}(e^{\alpha_1(\tau-D)}-e^{\alpha_1 t_1}).
    \end{aligned}
\end{equation}
We can now uncondition again on $T_{i-1,i}$ and get the \gls{pdf} of the \gls{paoi} in case $\mathcal{D}$:
\begin{equation}
\begin{aligned}
p_{\Delta|\mathcal{D}}(\tau)=&\int_D^{\tau-2D}\frac{\lambda\mu_1}{ p(\mathcal{D})}P_W(\tau-t_1-2D)e^{-\mu_1(\tau-D)}(e^{\alpha_1(\tau-D-t_1)}-1) dt_1\\
=&\frac{\lambda\mu_1}{ p(\mathcal{D})} e^{-\mu_1(\tau-D)}\left(e^{\alpha_1 D}\int_0^{\tau-2D}P_W(w)e^{\alpha_1 w}dw -\int_0^{\tau-2D}P_W(w) dw\right)\\
=&\frac{\mu_1\lambda(1-\lambda D)}{p(\mathcal{D})}e^{-\mu_1(\tau-D)}\sum_{k=0}^{\left\lfloor\frac{\tau-2D}{D}\right\rfloor}\Bigg[\frac{1}{\lambda}-e^{\alpha_1(k+1)D}+\sum_{j=0}^k\Bigg(\frac{(\tau-(k+2)D)^j }{j!}\\
&\times\left((-\lambda)^{j-1}e^{\lambda(\tau-(k+2)D)}-\frac{\lambda^k(-1)^j e^{\mu_1(\tau-(k+2)D)}}{\mu_1^{k-j+1}}\right)\Bigg)\Bigg].
\end{aligned} 
\end{equation}
The overall \gls{paoi} \gls{pdf} is then given by using the previous results for the four cases in \mbox{Definition \ref{th:totalpdf}}.

\section{PAoI distribution for the $M/M/1$ -- $M/M/1$ tandem}\label{sec:mm1}

We now consider the $M/M/1$ -- $M/M/1$ tandem, which represents edge computing-enabled systems with stochastic computation times or communication relaying systems. To calculate (\ref{eq:paoi}), we divide the computation in 4 cases, as we did in the previous section.

\subsection{The packet is queued at both systems}
We first consider case $\mathcal{A}$, in which packet $i$ finds both systems busy, i.e., the $i$-th packet arrives before the departure of the $(i-1)$-th packet at each system. In this case, $\Omega_{i,1}>0\wedge\Omega_{i,2}>0$. As the conditioned \gls{pdf} of $\Omega_{i,j}$ was given in~\eqref{eq:nextwait_mm1}, and we know that $Y_{i,1}$ is independent from $T_{i-1,1}$, as is $S_{i,1}$ from $T_{i-1,2}$, the probability of this case $p(\mathcal{A})$ is given by:
\begin{equation}
\begin{aligned}
 p(\mathcal{A})&=p(\Omega_1>0)p(\Omega_2>0|\Omega_1>0)=\\
 &=\int\displaylimits_0^\infty \int\displaylimits_0^{t_1} p_{y_1}(y_1) p_{T_1}(t_1) dy_1 dt_1\int\displaylimits_0^\infty\int\displaylimits_0^{t_2} p_{S_1}(s_1) p_{T_2}(t_2) ds_1 dt_2 \\
 &=\frac{\lambda}{\mu_1+\alpha_2}.
 \end{aligned}
\end{equation}
We start from the conditioned distribution of the system time on $\Omega_1$, $\Omega_2$, and $S_1$, so $S_2$ is the only remaining random variable. In the following, the index $i$ of the packet is omitted where possible to simplify the notation:
\begin{equation}
\begin{aligned}
 p_{T|\Omega_1,\Omega_2,S_1,\mathcal{A}}\left(t|\omega_1,\omega_2,s_1\right)=\mu_2 e^{-\mu_2\left(t-\omega_1-s_1-\omega_2\right)} u(t-\omega_1-\omega_2-s_1).
 \end{aligned}
\end{equation}
We now uncondition on $\Omega_2$, and then on $S_1$, by using the law of total probability:
\begin{equation}
\begin{aligned}
 p_{T|\Omega_1,\mathcal{A}}\left(t|\omega_1\right)=&\int\displaylimits_{0}^{\mathclap{t-\omega_1}}p_{S_1}(s_1)\int\displaylimits_{0}^{\mathclap{t-s_1-\omega_1}}\frac{p_{\Omega_2|\Omega_1,S_1}(\omega_2|\omega_1,s_1)}{1-P_{\Omega_2|\Omega_1,S_1}(0|\omega_1,s_1)} p_{T|\Omega_1,\Omega_2,S_1,\mathcal{A}}  d\omega_2 ds_1\\
 %=&\int\displaylimits_{0}^{t-\omega_1}\frac{\alpha_2}{\rho_2}\left(e^{-\alpha_2(t-\omega_1)}-e^{-\mu_2(t-\omega_1)+\lambda s_1}\right)\mu_1e^{-\mu_1s_1}ds_1\\
 =&\frac{\alpha_2\mu_2(\alpha_2+\mu_1)e^{\shortminus\alpha_2(t\shortminus\omega_1)}(\alpha_1+\lambda e^{\shortminus\mu_1(t\shortminus\omega_1)}-\mu_1e^{\shortminus\lambda(t\shortminus\omega_1)})}{\lambda\alpha_1\mu_1}.
 \end{aligned}
\end{equation}
The knowledge that we are in case $\mathcal{A}$ means that we have $\Omega_1>0$: the denominator in the first integral is the probability of this happening, which we need to account for to get the correct conditional probability. We then condition on $Y_1$ and uncondition on $\Omega_1$:
\begin{equation}
\begin{aligned}
 p_{T|Y_1,\mathcal{A}}\left(t|y_1\right)=&\int\displaylimits_{0}^{t} p_{T|\Omega_1,\mathcal{A}}\left(t|\omega_1\right) \frac{p_{\Omega_1|Y_1}(\omega_1|y_1)}{1-P_{\Omega_1|Y_1}(0|y_1)} d\omega_1\\
 =&\frac{\mu_2\alpha_2 e^{-\alpha_1 y_1}}{\lambda p(\mathcal{A})}\Bigg(\frac{\alpha_1(e^{-\alpha_1t}-e^{-\alpha_2t})}{(\mu_2-\mu_1)}+\frac{\lambda e^{-\alpha_1t}(1-e^{-\mu_2t})}{\mu_2}-\frac{\mu_1(e^{-\alpha_1t}-e^{-\mu_2t})}{\mu_2-\alpha_1}\Bigg).
\end{aligned}
\end{equation}
We can now derive the \gls{pdf} of the system time $T$:
\begin{equation}
\begin{aligned}
p_{T|\mathcal{A}}(t)=&\int_0^\infty p_{Y_1}(y_1) p_{T|Y_1,\mathcal{A}}(t|y_1) dy_1\\
=&\frac{\mu_2\alpha_2}{\mu_1 p(\mathcal{A})}\Bigg(\frac{\alpha_1(e^{-\alpha_1t}-e^{-\alpha_2t})}{(\mu_2-\mu_1)}
 +\frac{\lambda e^{-\alpha_1t}(1-e^{-\mu_2t})}{\mu_2}-\frac{\mu_1(e^{-\alpha_1t}-e^{-\mu_2t})}{\mu_2-\alpha_1}\Bigg).
\end{aligned}
\end{equation}
Finally, we get the \gls{pdf} of the \gls{paoi}, given by $T+Y_1$:
\begin{equation}
\begin{aligned}
p_{\Delta|\mathcal{A}}\left(\tau\right)=&\int_{0}^{\tau} p_{T|Y_1,\mathcal{A}}\left(t|\tau-t\right) p_{Y_1}(\tau-t) dt\\
=&\frac{\mu_1+\alpha_2}{\lambda}\bigg(\frac{\alpha_2\mu_1\mu_2(e^{-\mu_1\tau}-e^{-\mu_2\tau})}{(\mu_2-\mu_1)(\mu_2-\alpha_1)}-\lambda e^{-\mu_1\tau}(1-e^{-\alpha_2\tau})\\
&+\frac{\alpha_1\alpha_2\mu_2(e^{-\mu_1\tau}-e^{-\alpha_2\tau})}{(\mu_2-\mu_1)(\mu_1-\alpha_2)}+\frac{\alpha_1\mu_1\mu_2(e^{-\alpha_1\tau}-e^{-\mu_1\tau})}{(\mu_2-\mu_1)(\mu_2-\alpha_1)}\bigg).
\end{aligned}\label{eq:aoi_1a}
\end{equation}

\subsection{The packet is only queued at the first system}
We now consider case $\mathcal{B}$, in which the first system is busy but the second one is free when packet $i$ reaches it, i.e., the packet is not queued at the second system. We have $\Omega_{i,1}>0\wedge\Omega_{i,2}\leq0$, and this case happens with probability $p(\mathcal{B})$:
\begin{equation}
\begin{aligned}
 p(\mathcal{B})&=p(\Omega_1>0)p(\Omega_2\leq0|\Omega_1>0)\\
 &=\int\displaylimits_0^\infty \int\displaylimits_0^{t_1} p_{y_1}(y_1) p_{T_1}(t_1) dy_1 dt_1\int\displaylimits_0^\infty\int\displaylimits_{t_2}^\infty p_{S_1}(s_1) p_{T_2}(t_2) ds_1 dt_2 \\
 &=\frac{\lambda\alpha_2}{\mu_1(\mu_1+\alpha_2)}.
 \end{aligned}
\end{equation}
In this case, the system time \gls{pdf} is independent of $\Omega_2$, and we can just give the conditioned \gls{pdf} as:
\begin{equation}
\begin{aligned}
 p_{T|\Omega_1,S_1,\mathcal{B}}\left(t|\omega_1,s_1\right)=&\mu_2 e^{-\mu_2(t-\omega_1-s_1)}(1-e^{-\alpha_2s_1})u(t-\omega_1-s_1).
\end{aligned}
\end{equation}
As for case $\mathcal{A}$, we condition on $Y_1$ and uncondition on $S_1$ and $\Omega_1$:
\begin{equation}
\begin{aligned}
 p_{T|Y_1,\mathcal{B}}\left(t|y_1\right)=\frac{\mu_1(\mu_1+\alpha_2) e^{-\alpha_1(y_1+t)}}{\alpha_2}\Bigg(1-e^{-\mu_2t} -\frac{\alpha_2\mu_2(1-e^{(\alpha_1-\mu_2)t})}{(\mu_2-\mu_1)(\mu_2-\alpha_1)}+\frac{\alpha_1\mu_2(1-e^{-\lambda t})}{\lambda(\mu_2-\mu_1)}\Bigg).
\end{aligned}
\end{equation}
From this result, we derive the conditioned \gls{pdf} of the system time $T$ for case $\mathcal{B}$:
\begin{equation}
\begin{aligned}
p_{T|\mathcal{B}}(t)=\frac{\lambda e^{-\alpha_1t}}{p(\mathcal{B})}\Bigg(1-e^{-\mu_2t}+\frac{\alpha_1\mu_2(1-e^{-\lambda t})}{\lambda(\mu_2-\mu_1)}-\frac{\alpha_2\mu_2\left(1-e^{-(\mu_2-\alpha_1)t}\right)}{(\mu_2-\mu_1)(\mu_2-\alpha_1)}\Bigg).
\end{aligned}
\end{equation}
We can now find the unconditioned \gls{pdf} of the \gls{paoi} for case $\mathcal{B}$:
\begin{equation}
\begin{aligned}
 p_{\Delta|\mathcal{B}}\left(\tau\right)=&\frac{\mu_1}{p(\mathcal{B})}(e^{-\alpha_1\tau}-e^{-\mu_1\tau})-\frac{\lambda\mu_1e^{-\mu_1\tau}(1-e^{-\alpha_2\tau})}{\alpha_2p(\mathcal{B})}+\frac{\alpha_1\mu_1\mu_2(e^{-\alpha_1\tau}-(1+\lambda\tau)e^{-\mu_1\tau})}{\lambda(\mu_2-\mu_1)p(\mathcal{B})}\\
 &+\frac{\mu_1\mu_2\alpha_2\left((e^{-\mu_1\tau}-e^{-\alpha_1\tau})(\mu_2-\mu_1)+\lambda(e^{-\mu_1\tau}-e^{-\mu_2\tau})\right)}{(\mu_2-\mu_1)^2(\mu_2-\alpha_1)p(\mathcal{B})}.
\end{aligned}\label{eq:aoi_1b}
\end{equation}

\subsection{The packet is only queued at the second system}
We can then consider case $\mathcal{C}$, in which the $i$-th packet does not experience any queuing at the first system, i.e., $\Omega_1\leq0$, but there is queuing in the second system, i.e., $\Omega_2>0$. The probability of a packet experiencing case $\mathcal{C}$ is given by:
\begin{equation}
\begin{aligned}
 p(\mathcal{C})&=p(\Omega_1\leq0,\Omega_2>0)\\
 &=\int\displaylimits_0^\infty \int\displaylimits_0^{y_1}\int\displaylimits_0^\infty p_{y_1}(y_1) p_{T_1}(t_1)p_{S_1}(s_1)\int\displaylimits_{\mathclap{s_1-t_1+y_1}}^\infty p_{T_2}(t_2) dt_2 ds_1 dt_1 dy_1\\
 &=\frac{\lambda}{\mu_2(\mu_1+\alpha_2)}.
 \end{aligned}
 \end{equation}
The conditioned \gls{pdf} of the system time is:
\begin{equation}
 p_{T|\Omega_1,\Omega_2,S_1,\mathcal{C}}\left(t|\omega_1,\omega_2,s_1\right)=\mu_2 e^{-\mu_2\left(t-s_1-\omega_2\right)}u(t-s_1-\omega_2).
\end{equation}
As in case $\mathcal{A}$, we condition on $Y_1$ and uncondition on $\Omega_2$, $S_1$, and $\Omega_1$:
\begin{equation}
\begin{aligned}
 p_{T|Y_1,\mathcal{C}}\left(t|y_1\right)=\frac{\mu_2\alpha_2e^{-\alpha_2t}(e^{-\alpha_1y_1}-e^{-\alpha_2y_1})}{\lambda(\mu_2-\mu_1)p(\mathcal{C})}\left(\alpha_1-\mu_1e^{-\lambda t}+\lambda e^{-\mu_1t}\right).
\end{aligned}
\end{equation}
We can now find the \gls{pdf} of the system delay:
\begin{equation}
\begin{aligned}
 p_{T|\mathcal{C}}(t)=\frac{\alpha_2e^{-\alpha_2t}\left(\alpha_1-\mu_1e^{-\lambda t}+\lambda e^{-\mu_1t}\right)}{\mu_1 p(\mathcal{C})}.
\end{aligned}
\end{equation}
The conditioned \gls{pdf} of the \gls{paoi} is then:
\begin{equation}
\begin{aligned}
 p_{\Delta|\mathcal{C}}(\tau)=&\frac{\mu_2}{(\mu_2-\mu_1)p(\mathcal{C})}\Bigg(\frac{\alpha_1\alpha_2(e^{-\mu_1\tau}-e^{-\alpha_2\tau})}{\alpha_2-\mu_1}+\lambda e^{-\mu_1\tau}-\frac{\mu_1\alpha_2(e^{-\mu_1\tau}-e^{-\mu_2\tau})}{\mu_2-\mu_1}\\&-\frac{\alpha_1\alpha_2(e^{-\alpha_2\tau}-e^{-\mu_2\tau})}{\lambda} -\lambda e^{-(\mu_1+\alpha_2)\tau}+\alpha_2\mu_1\tau e^{-\mu_2\tau}-\frac{\lambda\alpha_2e^{-\mu_2\tau}(1-e^{-\alpha_1\tau})}{\alpha_1}\Bigg).
 \end{aligned}\label{eq:aoi_2a}
\end{equation}

\subsection{The packet is not queued at either system}
Finally, we examine case $\mathcal{D}$, in which the packet experiences no queuing, i.e., $\Omega_{i,1}\leq0 \wedge \Omega_{i,2}\leq0$. This case happens with probability $p(\mathcal{D})$:
\begin{align}
 p(\mathcal{D})=p(\Omega_1\leq0,\Omega_2\leq0)=\frac{\alpha_1\mu_2(\mu_1+\alpha_2))-\lambda\mu_1}{\mu_1\mu_2(\mu_1+\alpha_2)}.
\end{align}
Since the system time probability is independent of $\Omega_2$, we can just give the conditioned system time \gls{pdf} as:
\begin{equation}
\begin{aligned}
 p_{T|\Omega_1,S_1,\mathcal{D}}\left(t|\omega_1,s_1\right)=\frac{\mu_1\mu_2 e^{-\mu_2(t-s_1)}(1-e^{-\alpha_2(s_1-\omega_1)})}{\alpha_1} u(t-s_1).
\end{aligned}
\end{equation}
We then condition on $Y_1$ and uncondition on $S_1$ and $\Omega_1$:
\begin{equation}
\begin{aligned}
 p_{T|Y_1,\mathcal{D}}\left(t|y_1\right)=\frac{\mu_1\mu_2(e^{-\mu_1t}(1-e^{-\alpha_1y_1})-e^{-\mu_2t}(1-e^{-\alpha_2y_1})+e^{-(\mu_2+\alpha_1)t}(e^{-\alpha_1y_1}-e^{-\alpha_2y_1}))}{(\mu_2-\mu_1)p(\mathcal{D})}.
\end{aligned}
\end{equation}
The \gls{pdf} of the system time is:
\begin{equation}
\begin{aligned}
 p_{T|\mathcal{D}}(t)=\frac{\mu_2(\mu_1-\lambda)e^{-\mu_1t}-\mu_1(\mu_2-\lambda)e^{-\mu_2t}}{\lambda(\mu_2-\mu_1)p(\mathcal{D})}+\frac{\lambda e^{-(\mu_2+\alpha_1)t}}{p(\mathcal{D})}.
\end{aligned}
\end{equation}
We can now find the \gls{pdf} of the \gls{paoi} in case $\mathcal{D}$:
\begin{equation}
\begin{aligned}
 p_{\Delta|\mathcal{D}}(\tau)=\frac{\mu_1\mu_2\lambda}{p(\mathcal{D})}\Bigg(\frac{\tau(e^{-\mu_2\tau}-e^{-\mu_1\tau})}{\mu_2-\mu_1}+\frac{\alpha_1e^{-\mu_2\tau}-\alpha_2e^{-\mu_1\tau}}{\alpha_1\alpha_2(\mu_2-\mu_1)}+\frac{\left(e^{-\lambda\tau}+e^{-(\mu_1+\alpha_2)\tau}\right)}{\alpha_1\alpha_2}\Bigg).
\end{aligned}\label{eq:aoi_2b}
\end{equation}

As for the $M/M/1$ -- $M/D/1$ system, the total \gls{pdf} is given by Theorem \ref{th:totalpdf}.

\subsection{PAoI in the case with equal service rates}

In this subsection, we consider a special case in which the general formula of the \gls{paoi} \gls{pdf} is indeterminate, $\mu_1=\mu_2=\mu$. We follow the same steps as in the normal derivation. The other cases in which the general formula is indeterminate, $\mu_1=\alpha_2$ and $\mu_2=\alpha_1$, are not derived in this paper.

In case $\mathcal{A}$, i.e., for $\Omega_1>0\wedge\Omega_2>0$, we have:
\begin{equation}
p(\mathcal{A})=\frac{\lambda}{\mu+\alpha}.
\end{equation}
Following the same steps as in the general case, we get:
\begin{equation}
\begin{aligned}
p_{\Delta|\mathcal{A}}\left(\tau\right)=\frac{e^{-\mu\tau}}{p(\mathcal{A})}\left(\mu(e^{\lambda\tau}-1)+\lambda(e^{-\alpha\tau}-e^{\lambda\tau})+\frac{\mu\alpha(\alpha+\mu)(1-e^{\lambda\tau})+\mu\alpha\lambda\tau(\alpha e^{\lambda\tau}+\mu)}{\lambda^2}\right).
\end{aligned}
\end{equation}

The probability of being in case $\mathcal{B}$, i.e., $\Omega_1>0\wedge\Omega_2\leq0$, is:
\begin{equation}
    p(\mathcal{B})=\frac{\lambda\alpha}{\mu(\mu+\alpha)}.
\end{equation}
The conditioned \gls{pdf} of the \gls{paoi} is then:
\begin{equation}
\begin{aligned}
 p_{\Delta|\mathcal{B}}\left(\tau\right)=\frac{\mu e^{-\mu\tau}}{p(\mathcal{B})}\Bigg(\frac{\alpha^2(e^{\lambda\tau}-1)}{\lambda^2}-
 \frac{\lambda(1-e^{-\alpha\tau})}{\alpha}-\frac{\mu(\alpha\lambda\tau^2+(\alpha-\lambda)\tau)}{2\lambda}\Bigg).
\end{aligned}
\end{equation}
In this case, as the system is entirely symmetrical and both queues are $M/M/1$, it is also time reversible, making case $\mathcal{B}$ equivalent to case $\mathcal{C}$ in reverse.
The probability of being in case $\mathcal{C}$, i.e., $\Omega_1\leq0\wedge\Omega_2>0$, and the conditional \gls{pdf} of the \gls{paoi} are then the same as in case $\mathcal{B}$:
\begin{equation}
\begin{aligned}
p(\mathcal{C})=&p(\mathcal{B})
 p_{\Delta|\mathcal{C}}\left(\tau\right)=&p_{\Delta|\mathcal{B}}\left(\tau\right).
\end{aligned}
\end{equation}
Finally, we look at case $\mathcal{D}$, in which both systems are free:
\begin{equation}
 p(\mathcal{D})=\frac{\alpha(\mu+\alpha)-\lambda}{\mu(\mu+\alpha)}.
\end{equation}
We then have the conditioned \gls{pdf} of the \gls{paoi}:
\begin{equation}
 p_{\Delta|\mathcal{D}}\left(\tau\right)=\frac{\mu^2\lambda e^{-\mu\tau}(2\text{cosh}(\alpha\tau)-\alpha^2\tau^2-2)}{\alpha^2 p(\mathcal{D})}.
\end{equation}
As in the general case, the overall \gls{paoi} is given by Theorem~\ref{th:totalpdf}.

\section{Simulation results}\label{sec:results}

We compared the results of our analysis with a Monte Carlo simulation, transmitting $N=10^7$ packets and computing the system delay and \gls{paoi} for each. The initial stages of each simulation were discarded, removing $N_0=1000$ packets to ensure that the system had reached a steady state. We also divided the packets in the four cases, depending on the queuing they experienced at each system. As the derivation of the \gls{paoi} distribution does not involve any approximations, the simulation results should perfectly match the theoretical curves. The Monte Carlo simulation consisted of a single episode, with all packets being transmitted one after the other. The simulation parameters are listed in Table~\ref{tab:parameters}, and are used for all plots, unless otherwise specified.

\begin{table}[t]\centering
    \footnotesize
    \edit{
	\caption{Main simulation parameter values (unless otherwise specified).}
	\begin{tabular}[c]{ccl}
		\toprule
		Parameter & Value & Description \\
		\midrule
        $\lambda$ & 0.5 & Packet generation rate\\
        $\mu_1$ & 1 & Service rate of $M/M/1$ system 1\\
        $\mu_2$ & 1.25 & Service rate of $M/M/1$ system 2\\
        $D$ & 0.8 & Service time of the $M/D/1$ system\\
        $N$ & $10^7$ & Number of simulated packets\\
        $N_0$ & 1000 & Initial transition \\
        \bottomrule
	\end{tabular}
	\label{tab:parameters}}
\end{table}

\subsection{$M/M/1$ -- $M/D/1$ tandem}

We first consider the tandem in which the first queue is $M/M/1$, while the second is $M/D/1$. We note that, in all the following figures, the simulation results match the theoretically derived curves, showing the soundness of our calculations.

Although not shown, the system time has the expected behavior: it is highest in case $\mathcal{A}$ when the packet is queued at both, and lowest in case $\mathcal{D}$, in which both systems are free. However, the \gls{paoi} shows a different trend. While system time increases monotonically with the traffic load, the \gls{paoi} is the combination of the system and interarrival time: at one extreme, when the system has very low traffic, it is dominated by the interarrival time, while at the other, it is dominated by the system time. The optimal setting to minimize \gls{paoi} is somewhere in the middle, striking a balance between the two causes of age. Furthermore, deterministic service reduces uncertainty, particularly when traffic is high and queuing is the main cause of ageing.

\begin{figure}[!t]
 \centering
 \includegraphics{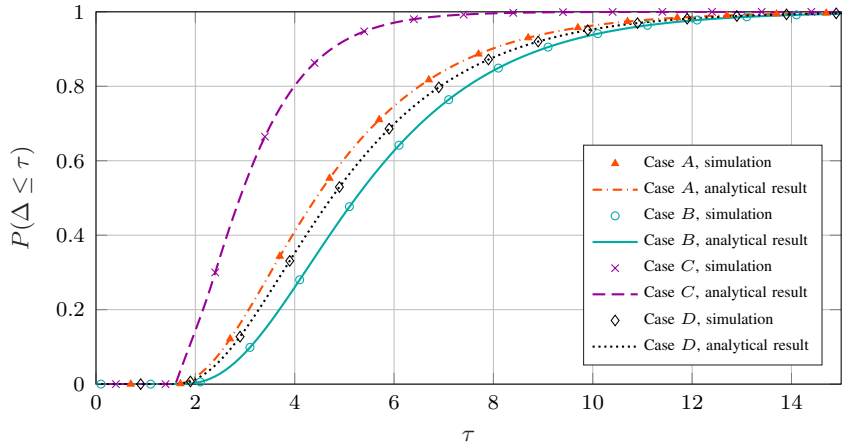}
 \vspace{-0.3cm}
 \caption{\gls{cdf} of the \gls{paoi} $\Delta$ for the $M/M/1$ -- $M/D/1$ tandem in the four subcases.}
   \vspace{-0.35cm}
 \label{fig:paoi_cases_md1}
\end{figure}

Fig.~\ref{fig:paoi_cases_md1} shows the \gls{cdf} of the \gls{paoi} in the four cases. It is interesting to note that the \gls{paoi} is never smaller than $2D$, as even serving packets instantaneously at the first system would still lead to a minimum delay: once packet $i-1$ is generated, it needs at least a time $D$ to get through the system because of the $M/D/1$ queue, and even if packet $i$ is generated right after it it needs another $D$ to be served by the edge node, leading to a minimum age of $2D$. The \gls{paoi} is far smaller in case $\mathcal{C}$, i.e., when the packet is queued only at the $M/D/1$ system, as the queue will often be short and is guaranteed to empty in a limited time. Cases $\mathcal{A}$ and $\mathcal{B}$ show a far worse performance, because of the first system's lower service rate ($D=0.8$ corresponds to a rate of 1.25) and of its exponential system time distribution. If the packet is not queued at either system (case $\mathcal{D}$), the \gls{paoi} is dominated by the interarrival time between the two packets, leading to higher ages.

\begin{figure}[t]
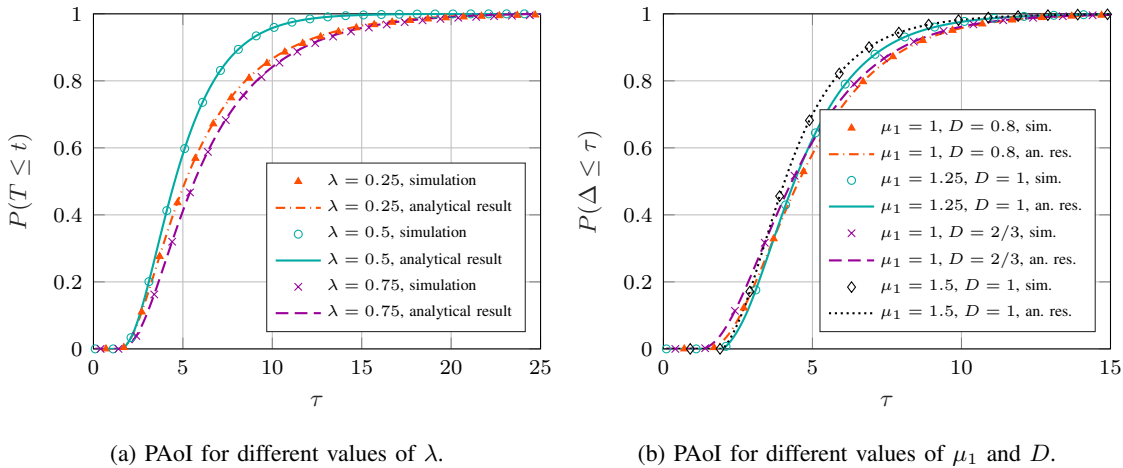

    \centering
 	\begin{subfigure}[b]{0.45\linewidth}
        \includegraphics{./aoi_lambda_md1.tex}
        \caption{\gls{paoi}  for different values of $\lambda$.}
        \label{fig:paoi_lambdas_md1}
     \end{subfigure}
     \begin{subfigure}[b]{0.45\linewidth}
        \includegraphics{./aoi_mu_md1.tex}
        \caption{\gls{paoi} for different values of $\mu_1$ and $D$.}
        \label{fig:paoi_mus_md1}
     \end{subfigure}
    \caption{\gls{cdf} of the \gls{paoi} $\Delta$ for the $M/M/1$ -- $M/M/1$ tandem for different values of $\lambda$, $\mu_1$, and $D$.}%\BS{COMMENT in the text the percentile of the PAoI is 90}}
    \label{paoi_md1}
\end{figure}

e can also examine the \gls{paoi} as a function of the generation and service rates: Fig.~\ref{fig:paoi_lambdas_md1} shows the \gls{paoi} \gls{cdf} for different values of $\lambda$. We can observe that a value of $\lambda$ close to 0.5 leads to the lowest \gls{paoi}, as a high traffic load increases queuing times, while a lower load increases \gls{paoi} due to the longer interarrival times. Fig.~\ref{fig:paoi_mus_md1} shows what happens when the service rates of the two systems are flipped.

In the figure, we compare two pairs of curves (orange and cyan, violet and black): the orange and violet dashed lines correspond to systems where the second node is faster, the cyan solid line and the black dotted line have the $M/D/1$ queue as the bottleneck. As the system time in the $M/D/1$ queues is rarely very large (it would need a very large queue to be significant, as all packets have the same service time), while $M/M/1$ queues have an exponential system time distribution which can take larger values more often, placing the bottleneck on the $M/D/1$ system leads to a lower \gls{paoi} in the worst case (i.e., the larger percentiles), at the cost of a worse \gls{paoi} in favorable scenarios. The difference between the orange and black lines is larger than between the cyan and violet ones, as a larger difference between the rates of the two links leads to an increased importance of the bottleneck. As for the subcase analysis, the system time and \gls{paoi} from the Monte Carlo simulation follow the analytical curve perfectly. \edit{As Fig.~\ref{fig:paoi_perc_mus_md1} shows, improving the edge computing capabilities has diminishing returns, as the $M/M/1$ communication system becomes the bottleneck: even with a very low $D$, the \gls{paoi} cannot be reduced beyond a certain value without also improving the first system's capacity.}

\begin{figure}[t]
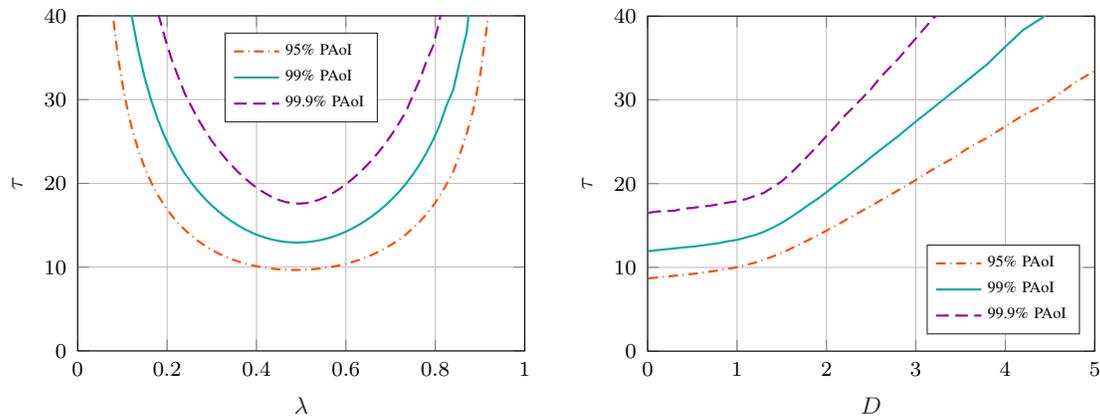

    \centering
 	\begin{subfigure}[b]{0.45\linewidth}
        \includegraphics{./aoi_percentiles_md1.tex}
        \caption{Percentiles for different values of $\lambda$ with $D=0.8$.}
        \label{fig:paoi_perc_md1}
     \end{subfigure}
     \begin{subfigure}[b]{0.45\linewidth}
        \includegraphics{./aoi_perc_mu_md1.tex}
        \caption{\edit{Percentiles for different values of $D$ with optimal $\lambda$.}}
        \label{fig:paoi_perc_mus_md1}
     \end{subfigure}
    \caption{Tail of the \gls{cdf} for the $M/M/1$ -- $M/D/1$ tandem with $\mu_1=1$.}%\BS{COMMENT in the text the percentile of the PAoI is 90}}
    \label{paoi_tail_md1}
\end{figure}

We also make a worst-case analysis as a function of $\lambda$: Fig.~\ref{fig:paoi_perc_md1} shows the 95th, 99th and 99.9th percentiles of the \gls{paoi} for a system with $\mu_1=1$ and $D=0.8$. If the traffic is very high, the queuing time is the dominant factor, causing the worst-case \gls{paoi} to diverge. The same happens if the traffic is too low, as the interarrival times can be very large: in this case, the system will almost always be empty, but updates will be very rare. The best performance in terms of \gls{paoi} is close to the middle. Depending on the desired reliability, system designers should choose the range of $\lambda$ that fulfils the given percentile of the \gls{paoi} in the specific conditions they consider.

\subsection{$M/M/1$ -- $M/M/1$ tandem}

\begin{figure}[!t]
 \centering
 \includegraphics{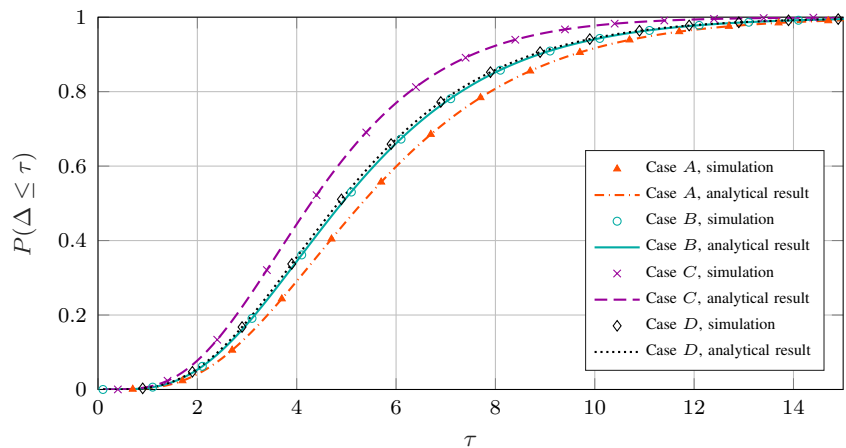}
 \vspace{-0.3cm}
 \caption{\gls{cdf} of the \gls{paoi} $\Delta$ for the $M/M/1$ -- $M/M/1$ tandem in the four subcases.}
   \vspace{-0.35cm}
 \label{fig:paoi_cases}
\end{figure}

\begin{figure}[t]
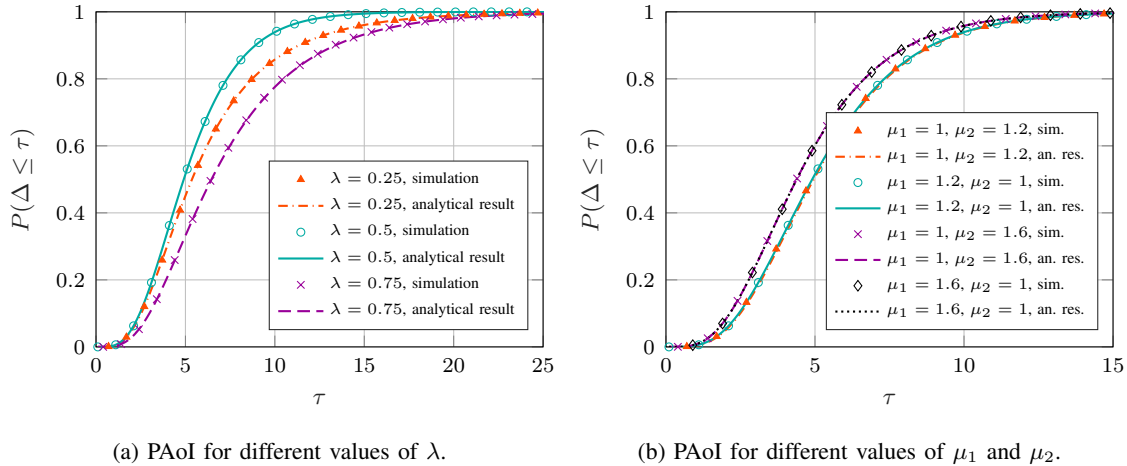
\
    \centering
 	\begin{subfigure}[b]{0.45\linewidth}
        \includegraphics{./aoi_general.tex}
        \caption{\gls{paoi}  for different values of $\lambda$.}
        \label{fig:paoi_lambdas}
     \end{subfigure}
     \begin{subfigure}[b]{0.45\linewidth}
        \includegraphics{./aoi_mu.tex}
        \caption{\gls{paoi} for different values of $\mu_1$ and $\mu_2$.}
        \label{fig:paoi_mus}
     \end{subfigure}
    \caption{\gls{cdf} of the \gls{paoi} $\Delta$ for the $M/M/1$ -- $M/M/1$ tandem for different values of $\lambda$, $\mu_1$, and $\mu_2$.}%\BS{COMMENT in the text the percentile of the PAoI is 90}}
    \label{paoi_mm1}
\end{figure}

Fig.~\ref{fig:paoi_cases} shows the \gls{paoi} \gls{cdf} in the four subcases for $\lambda=0.5$, $\mu_1=1$, and $\mu_2=1.2$. Interestingly, there is no minimum delay and the \gls{cdf} starts in zero, unlike in the $M/M/1$--$M/D/1$ case. The \gls{paoi}
is the lowest in case $\mathcal{C}$, and almost identical in cases $B$ and $D$. This is due to the effect of the interarrival times on the \gls{paoi}, as case $\mathcal{D}$ usually means that the instantaneous load of the system is low and packets are far apart, increasing the \gls{paoi}. In case $\mathcal{C}$, the faster system is busy and the bottleneck is empty. Intuitively, this can reduce age, as the second system will probably be able to serve packets fast enough, but at the same time the instantaneous load will be high enough to avoid having a strong impact on the age.

We can now examine the \gls{paoi} \glspl{cdf} for different values of $\lambda$: the system time is always higher for higher values of $\lambda$, as it depends on the traffic. The same is not true for the \gls{paoi}, as Fig.~\ref{fig:paoi_lambdas} shows: as for the $M/M/1$ -- $M/D/1$ tandem, the \gls{paoi} is lowest for $\lambda=0.5$, as the high interarrival time becomes the dominant factor for $\lambda=0.25$. Interestingly, the gap between the system with $\lambda=0.25$ and the one with $\lambda=0.75$ is wider: the $M/D/1$ system is better able to handle high load situations, as having a high system time is far rarer than in the $M/M/1$. On the other hand, the values of $\mu_1$ and $\mu_2$ also have an important effect, as Fig.~\ref{fig:paoi_mus} shows: while the bottleneck always has a service rate 1, changing the service rate of the other link and even switching the two can have an impact on the \gls{paoi}. Naturally, increasing the rate of the other link from 1.2 to 1.6 slightly reduces the \gls{paoi}, but we note that for both values, having the first or second system as the bottleneck has a negligible effect on performance. Unlike in the $M/M/1$ -- $M/D/1$ tandem, the location of the bottleneck in the tandem seems to have a very small influence on the distribution of the \gls{paoi}, as both queues have the same type of service time distribution.

%Finally, 
Fig.~\ref{fig:paoi_perc} shows how the worst-case \gls{paoi}, measured using the 95th, 99th and 99.9th percentiles, changes as a function of $\lambda$. The figure shows that the trends for the two systems are similar: the $M/M/1$ -- $M/D/1$ queue can handle a high load slightly better, and its minimum \gls{paoi} is lower by 5-10\% at all the considered percentiles. In both cases, the optimal $\lambda$ is between 0.45 and 0.5 for all the three percentiles. \edit{As for the $M/M/1$ -- $M/D/1$ tandem, increasing the rate of the second system gives diminishing returns, and any increase past $\mu_2=2$ has negligible benefits if $\mu_1=1$, as Fig.~\ref{fig:paoi_perc_mus_mm1} shows. The horizontal axis in the figure shows the inverse of $\mu_2$ to provide a visual comparison with Fig.~\ref{fig:paoi_perc_md1}: the tail of the \gls{paoi} distribution clearly has higher values in the system with two $M/M/1$ queues, as we would expect due to the additional randomness in the service time of the edge computing node.}

\begin{figure}[t]
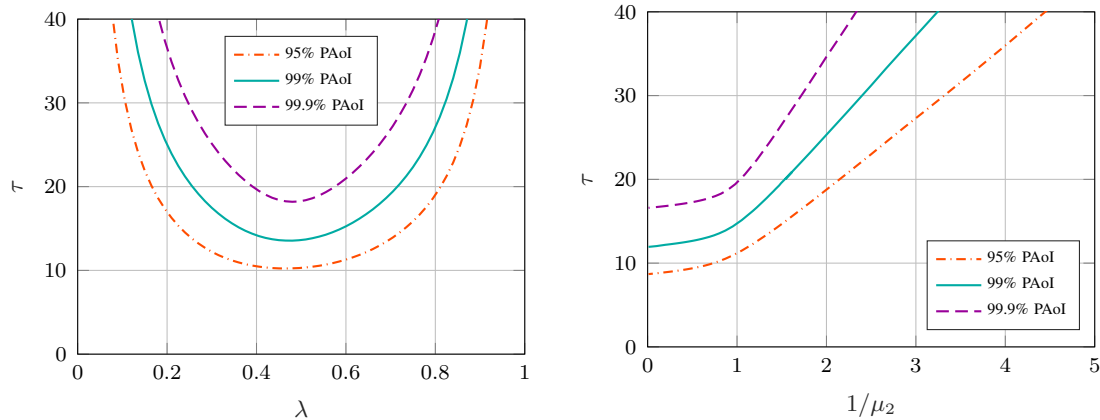

    \centering
 	\begin{subfigure}[b]{0.45\linewidth}
        \includegraphics{./aoi_percentiles.tex}
        \caption{Percentiles for different values of $\lambda$ with $\mu_2=1.25$.}
        \label{fig:paoi_perc}
     \end{subfigure}
     \begin{subfigure}[b]{0.45\linewidth}
        \includegraphics{./aoi_perc_mu_mm1.tex}
        \caption{\edit{Percentiles for different values of $\mu_2$ with optimal $\lambda$.}}
        \label{fig:paoi_perc_mus_mm1}
     \end{subfigure}
    \caption{Tail of the \gls{cdf} for the $M/M/1$ -- $M/M/1$ tandem with $\mu_1=1$.}%\BS{COMMENT in the text the percentile of the PAoI is 90}}
    \label{paoi_tail_mm1}
\end{figure}

Finally, Fig.~\ref{fig:paoi_mm1_md1} shows a comparison of the two kinds of system: as expected, the $M/M/1$ -- $M/D/1$ tandem has a higher \gls{paoi} in the best-case scenario, as it has a hard minimum of $2D$, but quickly becomes better at the higher percentiles. This difference is more pronounced if the second system is slower, with the $M/M/1$ -- $M/D/1$ tandem with $D=1$ having a better \gls{paoi} than the $M/M/1$ -- $M/M/1$ one with $\mu_2=1.25$ for percentiles above the 60th. As the plot shows, the two types of tandem are substantially equivalent if the computation is much faster than the communication. We also show the average \gls{paoi} for the two systems in Fig.~\ref{fig:avpaoi_mm1_md1}: this is computable from the \gls{pdf} we derived in this paper, but simpler formulas were derived in the literature~\cite{xu2019optimizing,kam2018age}. As expected, the $M/M/1$--$M/D/1$ system is better able to deal with high traffic load than the  $M/M/1$--$M/M/1$. It is interesting to note that the $M/M/1$--$M/D/1$ queue has a larger gain in the average \gls{paoi} than at higher percentiles, even for values of $\lambda$ between 0.45 and 0.5.

\begin{figure}[!t]
\centering
\includegraphics{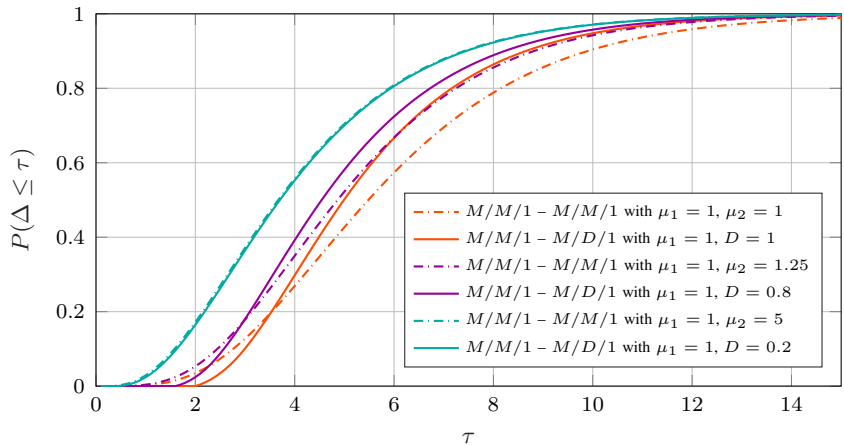}
 \vspace{-0.3cm}
 \caption{\gls{cdf} of the \gls{paoi} $\Delta$ for the two systems (analytical results) for different values of $\mu_1$ and $D$.}
  \vspace{-0.4cm}
 \label{fig:paoi_mm1_md1}
\end{figure}

\begin{figure}[!t]
\centering
\includegraphics{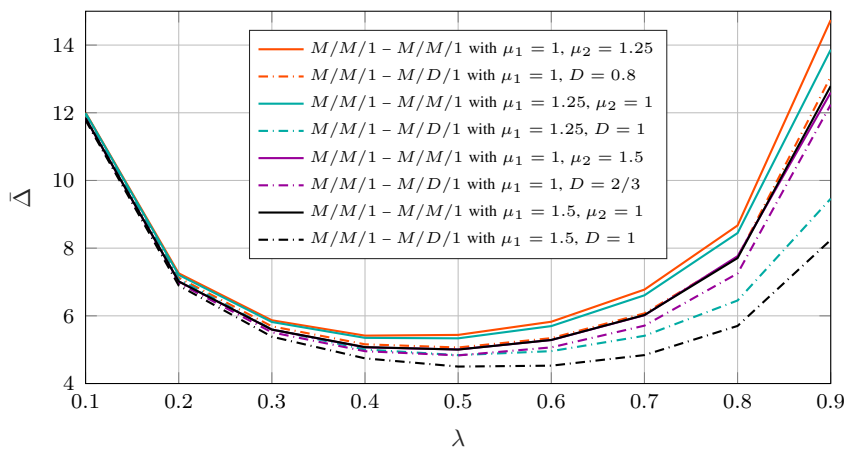}
 \vspace{-0.3cm}
 \caption{Average \gls{paoi} $\bar{\Delta}$ for the two systems (analytical results) for different values of $\mu_1$ and $D$.}
  \vspace{-0.4cm}
 \label{fig:avpaoi_mm1_md1}
\end{figure}

\section{Conclusions and future work}\label{sec:conclusions}

In this paper, we derived the \gls{pdf} of the \gls{paoi} for a tandem with an $M/M/1$ system followed by an $M/D/1$ and for one consisting of two $M/M/1$ queues. These new results can give more flexibility in the design of bounded \gls{aoi} systems, both for edge-enabled \gls{iot} where the wireless tranmission precedes the computation delay and in other relay applications. The results are derived for two nodes, but the procedure is generic for $K$ $M/M/1$ nodes, potentially followed by an $M/D/1$ system. If the $M/D/1$ system is not the last, the analytical derivation of the \gls{paoi} becomes intractable, as its departure process is non-Markovian. 

\edit{The optimization of the two systems can be a complex operation to perform in real time, and finding the complete \gls{pdf} analytically might not be manageable in more complex cases with channel errors and multiple sources. However, approximations of the optimal policies can be found by applying deep learning techniques~\cite{she2020tutorial}, paying particular attention to the uncertainty in the model parameters in order to guarantee reliability even in realistic conditions~\cite{angjelichinoski2019statistical}. The possibility to compute the \gls{pdf} and \gls{cdf} using the formulas can provide a fast way to generate samples for a deep learning system, which would then be able to generalize experience in previously unseen situations with a lower computational cost: there is an extensive body of work on similar approaches, which use learning to quickly find a solution to optimization problems with computationally expensive cost functions.} \edit{Aside from the possibilities offered by deep learning, a} possible avenue of future work is the introduction of multiple independent sources in the system, possibly with different priorities. The extension of the system to longer or more complex queuing networks is also a possibility, but the complexity of the derivation might make the analytical results unwieldy. The inclusion of error-prone links, with packets being randomly dropped, is another interesting extension. Finally, the optimization of $\lambda$ in all these scenarios might be a practical application of the theoretical derivations, yielding usable scheduling policies.

\section*{Acknowledgment}

\edit{This work has been in part supported by the Danish Council for Independent Research, Grant  8022-00284B SEMIOTIC, and by the European Union’s Horizon 2020 research and innovation program as part of the IntellIoT project, which received funding under grant agreement 957218.}

% or


\ifCLASSOPTIONcaptionsoff
  \newpage
\fi



\bibliographystyle{IEEEtran}
\bibliography{bibliography.bib}


%\end{IEEEbiography}
% \end{IEEEbiography}
% % insert where needed to balance the two columns on the last page with
% %\newpage


%\vfill


\end{document}